\documentclass[a4paper,11pt]{article}
\usepackage{a4wide}
\usepackage{amsmath,amsthm}
\usepackage{csquotes}
\usepackage{anyfontsize}
\usepackage{a4wide}
\usepackage[type1,sb]{libertine}
\usepackage[libertine]{newtxmath}
\usepackage{inconsolata}
\usepackage[scr=rsfso]{mathalfa}
\usepackage[UKenglish]{babel}
\usepackage{microtype}
\usepackage{interval}
\intervalconfig{soft open fences}
\usepackage{thm-restate}
\usepackage{complexity}
\usepackage{comment}
\usepackage{hyperref} 
\usepackage{xcolor}

\definecolor{myBlue}{RGB}{27, 82, 140}
\definecolor{myBrightBlue}{RGB}{24, 107, 196}
\definecolor{myGreen}{RGB}{82, 140, 27}
\definecolor{myBrightGreen}{RGB}{107, 196, 24}
\definecolor{myRed}{RGB}{140, 27, 82}
\definecolor{myBrightRed}{RGB}{196, 24, 107}

\hypersetup{colorlinks,allcolors=myBlue}

\usepackage{todonotes} 
\usepackage{optidef} 
\usepackage{tikz}
\usepackage{pgfplots}
\usetikzlibrary{calc}
\usepackage[nameinlink]{cleveref}
\usepackage{bm}
\usepackage{enumitem}

\newcommand\B[1]{{\textstyle\binom{#1}{\lfloor #1/2 \rfloor}}}
\newcommand{\vx}{\ensuremath{\mathbf{x}}}
\newcommand{\ve}{\ensuremath{\mathbf{e}}}
\newcommand{\vy}{\ensuremath{\mathbf{y}}}

\newcommand{\va}{\ensuremath{\mathbf{a}}}

\renewcommand{\epsilon}{\varepsilon}
\renewcommand{\phi}{\varphi}

\newcommand{\ignore}[1]{}

\newlist{algoEnum}{enumerate}{10}
\setlist[algoEnum, 1]{label={\arabic*.}}
\setlist[algoEnum, 2]{label={(\roman*)}}

\creflabelformat{algoEnumi}{#2#1#3}
\creflabelformat{algoEnumii}{#2#1#3}

\crefname{algoEnumi}{step}{steps}
\crefname{algoEnumii}{case}{cases}

\pgfplotsset{compat=1.18}
\usepgfplotslibrary{fillbetween}
\usetikzlibrary{calc}
\usetikzlibrary{patterns}

\renewcommand{\G}{\ensuremath{\mathfrak{G}}}

\DeclareMathOperator{\NBin}{NBin}

\DeclareMathOperator{\elin}{\textnormal{\textsf{E3Lin}}}
\DeclareMathOperator{\belin}{\textnormal{\textsf{BalancedE3Lin}}}

\DeclareMathOperator{\maxPCSP}{\textup{MaxPCSP}}

\DeclareMathOperator{\maxCSP}{\textup{MaxCSP}}

\DeclareMathOperator{\CSP}{\textup{CSP}}
\DeclareMathOperator{\PCSP}{\textup{PCSP}}

\newcommand{\N}{\ensuremath{\mathbb{N}}}
\renewcommand{\R}{\ensuremath{\mathbb{R}}}

\newcommand{\norm}[1]{\left \lVert #1 \right \rVert}

\newclass\UGC{UGC}

\theoremstyle{plain}
\newtheorem{theorem}{Theorem}
\newtheorem{lemma}[theorem]{Lemma}
\newtheorem*{lemma*}{Lemma}

\newtheorem*{proposition*}{Proposition}
\newtheorem{corollary}[theorem]{Corollary}
\newtheorem*{corollary*}{Corollary}
\newtheorem{conjecture}[theorem]{Conjecture}
\theoremstyle{definition}
\newtheorem{definition}[theorem]{Definition}

\begin{document}

\title{A Dichotomy for Maximum PCSPs on Graphs\thanks{An extended abstract of a part of this work appeared in Proceedings of ICALP\,(A) 2025~\cite{NZ25:icalp}. This work was supported by UKRI EP/X024431/1 and Clarendon Fund Scholarship. Work done while Tamio-Vesa Nakajima was at the University of Oxford. For the purpose of Open Access, the authors have applied a CC BY public copyright licence to any Author Accepted Manuscript version arising from this submission. All data is provided in full in the results section of this paper.}}

\author{Tamio-Vesa Nakajima\\
UUniversity of Marburg\\
\url{nakajima@uni-marburg.de}
\and
Stanislav \v{Z}ivn\'y\\
University of Oxford\\
\url{standa.zivny@cs.ox.ac.uk}
}

\date{} 

\maketitle

\begin{abstract}
Fix two non-empty loopless graphs $G$ and $H$ such that $G$ maps homomorphically to $H$. The \emph{Maximum Promise Constraint Satisfaction Problem} parameterised by $G$ and $H$ is the following computational problem, denoted by $\maxPCSP(G, H)$: Given an input (multi)graph $X$ that admits a map to $G$ preserving a $\rho$-fraction of the edges, find a map from $X$ to $H$ that preserves a $\rho$-fraction of the edges. As our main result, we give a complete classification of this problem under Khot's Unique Games Conjecture: The only tractable cases are when $G$ is bipartite and $H$ contains a triangle. 

Along the way, we establish several results, including an efficient approximation algorithm for the following problem: Given a (multi)graph $X$ which contains a bipartite subgraph with $\rho$ edges, what is the largest triangle-free subgraph of $X$ that can be found efficiently? We present an SDP-based algorithm that finds one with at least $0.8823 \rho$ edges, thus improving on the subgraph with $0.878 \rho$ edges  obtained by the classic Max-Cut algorithm of Goemans and Williamson.
\end{abstract}

\section{Introduction}

Given two undirected graphs\footnote{All graphs in this article are loopless and
non-empty, meaning having at least one edge (and thus at least two vertices).}
$G$ and $H$, a \emph{homomorphism} from $G$ to $H$ is an edge preserving map $h$
from $V(G)$ to $V(H)$; that is, if $(u,v)\in E(G)$ then $(h(u),h(v))\in
E(H)$. A classic result of Hell and Ne\v{s}et\v{r}il established a computational
\emph{dichotomy} for the so-called $H$-colouring
problem~\cite{Hell90:h-coloring}, for a fixed graph $H$: if $H$ is bipartite then deciding whether an
input graph $G$ is homomorphic to $H$ is solvable in polynomial time, and for
every other $H$ this problem is NP-complete. Going beyond graphs, Feder and
Vardi conjectured that a similar dichotomy holds for all finite relational
structures~\cite{FederVardiConjecture}, not only for graphs and for relational structures over
the Boolean domain~\cite{Schaefer78:stoc}.\footnote{We will not need relational
structures, but, intuitively, one should think of them as generalisations of
(hyper)graphs in which one is given a ground set and a collection of relations
on the ground set.} Bulatov~\cite{Bulatov17:focs} and, independently,
Zhuk~\cite{Zhuk20:jacm}, confirmed the tractability part of the dichotomy, which
together with the NP-hardness part~\cite{Bulatov05:classifying}, answered the
Feder-Vardi conjecture in the affirmative. The homomorphism problem is also
known as the \emph{constraint satisfaction problem} (CSP)~\cite{Jeavons98:algebraic}. CSPs can be
equivalently defined as problems seeking an assignment of domain values to the
given variables subject to the given constraints. The fixed target structure in
the homomorphism problem corresponds to the set of allowed (domain) values and
the set of allowed relations in the constraints. Homomorphisms taking edges to
edges corresponds to insuring that all constraints are satisfied. Concrete examples of CSPs include
solvability of linear equations over finite fields and variants of (hyper)graph
colourings.

A well-studied line of work focuses on \emph{approximability} of
CSPs~\cite{KSTW00,Trevisan00:sicomp}. A classic example here is
the Max-Cut problem. In Max-Cut, the variables correspond to the
vertices of the input graph, the values are just $0$ and $1$ (corresponding to
the two sides of a cut), and the constraints are binary disequalities associated
with the edges of the graph. Given a CSP, the computational task could be to
find a solution maximising the number of satisfied constraints as in Max-Cut, or
finding a (perfect) solution satisfying all constraints as discussed in the
previous paragraph. 

A \emph{promise} CSP (PCSP) is a CSP in which each
constraint comes in two forms, a strong one and a weak one. The promise is that
a solution exists using the strong versions of the constraints, while the
(possibly easier) task is to find a solution using the weak constraints. A
recent line of work by Austrin, Guruswami, and H{\aa}stad~\cite{AGH17},
Brakensiek and Guruswami~\cite{BG21:sicomp}, and Barto, Bul\'in, Krokhin, and
Opr\v{s}al~\cite{BBKO21} initiated a systematic study of PCSPs with perfect completeness,
i.e., finding a solution satisfying all weak constraints given the promise that
a solution satisfying all strong constraints exists. Canonical examples include 
approximate graph~\cite{GJ76} and hypergraph~\cite{DRS05,ABP20,Barto21:stacs} colouring problems,
e.g., finding a 5-colouring of a given 3-colourable graph~\cite{KOWZ23}.
PCSPs are a vast generalisation of CSPs and their complexity is not well
understood, not even on the Boolean domain~\cite{Ficak19:icalp,BGS23} or for
graphs. In particular, Brakensiek and Guruswami conjectured that only bipartite
graphs lead to tractable PCSPs on graphs~\cite{BG21:sicomp} (cf.~\Cref{conjBG}
in~\Cref{sec:results} for a precise statement). Resolving their 
conjecture would in particular include resolving the notoriously difficult approximate graph
colouring problem, cf.~\cite{Avvakumo25:stoc} for exciting recent progress.

\medskip

In this work, we will focus on the \emph{approximability of maximisation PCSPs}. The ultimate goal is to understand the precise
approximation factor for all $\maxPCSP$s, and thus identify where the transition
from tractability to intractability occurs. This is an ambitious, long-term goal
that would encompass many existing fundamental results.

An example of a $\maxPCSP$ is the following problem. Given a graph $G$ that
admits a $2$-colouring of the vertices with a $\rho$-fraction of the edges
coloured properly, find a $3$-colouring of $G$ with an $\alpha\rho$-fraction of
the edges coloured properly, where $0<\alpha\leq 1$ is the approximation factor.
As one of the results in the present paper, we will show that there is a
1-approximation algorithm; i.e., given a graph with a 2-colouring with
$\rho$-fraction of non-monochromatic edges, one can efficiently find a
3-colouring of the graph with the same fraction of non-monochromatic edges. 

As our main result, we will establish a dichotomy for 1-approximation of graph
PCSPs under Khot's Unique Games Conjecture (\UGC)~\cite{Khot02stoc}.

\begin{restatable}{theorem}{UGCdichotomy}\label{thm:UGCdichotomy}
  Let $G$ and $H$ be two fixed graphs such that there is a homomorphism from $G$ to $H$. 
  If $G$ is bipartite and $H$ contains a triangle then $\maxPCSP(G,H)$ is 1-approximable. Otherwise, 1-approximation of $\maxPCSP(G,H)$ is \NP-hard assuming the \UGC.
\end{restatable}

Along the way to prove~\Cref{thm:UGCdichotomy}, we will design two efficient
approximation algorithms. We shall discuss one of them briefly here, with an
overview of both algorithms and all results in~\Cref{sec:results}.

Given an undirected (multi)graph $G$, what is the bipartite subgraph of $G$ with
the most edges? This is nothing but the already mentioned Max-Cut problem, one of the most fundamental problems in computer science.
Max-Cut was among the 21 problems shown to be \NP-hard by Karp~\cite{Karp1972}.
Papadimitriou and Yannakakis showed that Max-Cut is
\APX-hard~\cite{PapadimitriouY91} and thus does not admit a polynomial-time
approximation scheme, unless $\P=\NP$. However, there are several simple
$0.5$-approximation algorithms. Goemans and Williamson used semidefinite
programming and randomised rounding to design a $0.878$-approximation
algorithm~\cite{GW95}. 
Khot, Kindler, Mossel, O'Donnell, and Oleszkiewicz established the optimality of
this algorithm~\cite{KKMO07,Mossel10:ann} under the \UGC.

\begin{center}\emph{What if the task is merely finding a large
  \underline{triangle-free} subgraph (rather than a bipartite one)?}
\end{center}

While the Goemans-Williamson algorithm can still be used, as one of our results
we design an algorithm with a better approximation guarantee: If $G$ contains a
bipartite subgraph with $\rho$ edges, our algorithm efficiently finds a
triangle-free subgraph of $G$ with $0.8823\rho$ edges.\footnote{We note that our
algorithm can be easily extended to the case where edges have positive weights.} Our
algorithm is a randomised combination of the Goemans-Williamson original
``random hyperplane algorithm'', and an algorithm that first selects ``long
edges'' (meaning edges for which the angle between the corresponding vectors
from the SDP solution is above a certain threshold) and then applies a random
hyperplane rounding, selecting ``shorter edges'' (still longer than some other
threshold). The probability of the biased coin that selects one of the two
algorithms depends on certain geometric quantities which guarantee that the
resulting subgraph is indeed triangle-free.
We complement our tractability result for this problem by showing that it is
$\NP$-hard to find a triangle-free subgraph with $(25 / 26 + \epsilon) \rho
\approx (0.961 + \epsilon) \rho$ edges. This result is obtained by a reduction
from H\aa{}stad's 3-bit PCP~\cite{Hastad01}.

\paragraph{Related work}

The notion of $\maxPCSP$s is a natural generalisation of the well-studied notion of $\maxCSP$s. 
For finite-domain $\maxCSP$s, it is known that a certain rounding of the basic SDP relaxation gives,
up to some $\epsilon$,
the \UGC-optimal approximation ratio (in time doubly exponential in $1 /
\epsilon$)~\cite{Raghavendra08:everycsp,Raghavendra09:focs}. 
However, the
Raghavendra-Steurer algorithm does not immediately give a 1-approximation
algorithm due to the above-mentioned $\epsilon$, even for $\maxCSP$s. Moreover,
that result is established only for finite-domain $\maxCSP$s. 
On other hand, our results include a 1-approximation for $\maxPCSP(K_2,
K_3)$, and an algorithm for infinite-domain structures, namely for
$\maxPCSP(K_2, \G_3)$, which captures the bipartite vs. triangle-free subgraph
discussed above (cf.~\Cref{sec:prelims} for a precise definition).

Approximation of concrete $\maxPCSP$s has been studied for decades, including
several papers on almost approximate graph
colouring~\cite{Engebretsen08:rsa,Dinur10:focs,Khot12:focs,Hecht23:approx},
approximate colouring~\cite{nz23:arxiv}, and promise
Max-3-LIN~\cite{blz25:icalp}. Our work initiates a systematic investigation,
giving a complete classification for 1-approximation of the graph case.

Recent work of Brakensiek, Guruswami, and Sandeep~\cite{BGS23:stoc} studied
robust approximation of $\maxPCSP$s; in particular, they state that
Raghavendra's above-mentioned theorem on approximate
$\maxCSP$s~\cite{Raghavendra08:everycsp} applies verbatim to $\maxPCSP$s. This
in combination with the work of Brown-Cohen and Raghavendra~\cite{BCR15} gives a
framework for studying approximation of $\maxPCSP$s. An alternative framework
for studying approximation of $\maxPCSP$s has recently been put forward by
Barto, Butti, Kazda, Viola, and \v{Z}ivn\'y~\cite{Barto24:lics-algebraic}. 

\paragraph{Paper organisation}

After defining $\maxPCSP$s formally and few other basic concepts
in~\Cref{sec:prelims}, we will state all our results precisely
in~\Cref{sec:results}. The rest of the paper is then split in different parts of
the proofs, including two tractability results in~\Cref{sec:tractability}
and~\Cref{sec:23} and hardness results in~\Cref{sec:hardness}
and~\Cref{sec:ugchardness}

\section{Preliminaries}\label{sec:prelims}

For two nonzero vectors $\vx, \vy\in \R^N$, we denote by $\angle(\vx, \vy)$ the angle between $\vx$ and $\vy$ in radians; i.e.,
$\angle(\vx, \vy) = \arccos\left( \frac{\vx \cdot \vy}{\norm{\vx}\norm{\vy}}\right)$.
The following useful fact is well-known, cf. \cite[Book \textsc{xi}, Proposition 21]{euclid}.

\begin{lemma}\label{lem:noTriangles}
    For any three nonzero vectors $\vx_1, \vx_2, \vx_3 \in \R^N$, we have
    $\angle(\vx_1, \vx_2) + \angle(\vx_2, \vx_3) + \angle(\vx_3, \vx_1) \leq 2\pi $.
\end{lemma}

\paragraph{Graphs and (partial) homomorphisms.}
All graphs will be nonempty, undirected and loopless but with possibly multiple edges.
Fix two  graphs $G = (V, E)$, $H = (U, F)$, with vertex sets $V, U$ and edge multisets $E, F$ respectively, as well as $\rho \in \N$. We say that there exists a 
\emph{partial homomorphism} of weight $\rho$ from $G$ to $H$, and write $G \xrightarrow{\rho} H$, if there exists a mapping $h : V \to U$ such that for at least $\rho$ edges $(x, y) \in E$ we have $(h(x), h(y)) \in F$. If $G \xrightarrow{|E|} H$, we say that there exists a homomorphism from $G$ to $H$ and write $G \to H$. (Note that for any $G, H, I$, if $G \xrightarrow{\rho} H \rightarrow I$  then $G \xrightarrow{\rho} I$.)

We denote by $K_2$ a clique on two vertices.
A partial homomorphism $h:G \xrightarrow{\rho}K_2$ represents a cut of weight $\rho$, namely the edges $(x,y)$ with $h(x)\neq h(y)$. Equivalently, it represents a bipartite subgraph of $G$ with weight $\rho$. We now introduce a graph that similarly captures triangle-free subgraphs.
Let $\G_3$ be the direct sum of all finite triangle-free graphs. In other words, for every finite triangle-free graph $G = (V, E)$, the graph $\G_3$ contains vertices $x_G$ for $x \in V$, and edges $(x_G, y_G)$ for $(x, y) \in E$. Then, for finite $G$, a partial homomorphism $h : G \xrightarrow{\rho} \G_3$ represents a triangle-free subgraph of $G$ with weight $\rho$: all the edges that connect vertices that are mapped by $h$ to neighbouring vertices in $\G_3$ form a triangle-free subgraph of $G$.

\paragraph{Maximum PCSPs.} Fix two (possibly infinite) graphs $G \rightarrow H$.
Then the \emph{maximum promise constraint satisfaction problem} ($\maxPCSP$) for
undirected graphs, denoted by $\maxPCSP(G, H)$, is defined as follows. In the
search version of the problem, we are given a (multi)graph $X$ such that $X
\xrightarrow{\rho} G$, and must find $h :X \xrightarrow{\rho} H$; this problem
can be approximated with the approximation ratio $\alpha$ if we can find $h :
X\xrightarrow{\lceil \alpha \rho\rceil} H$. In the decision version, we are
given a (multi)graph $X$ and a number $\rho \in \N$ and must output \textsc{Yes}
if $X \xrightarrow{\rho}G$, and \textsc{No} if not even $X \xrightarrow{\rho}H$.
This problem can be approximated with approximation ratio $\alpha$ if we can
decide between $X \xrightarrow{\rho}G$ and not even $X \xrightarrow{\lceil\alpha
\rho\rceil} H$.\footnote{In all cases, the \emph{maximum} $\rho$ for which $X
\xrightarrow{\rho} G$ is \emph{not} part of the input; for the decision version
in particular, the number $\rho$ we are given can be arbitrary, it does not
necessarily have any relationship with the graph $X$.}

In particular, approximating the problem $\maxPCSP(K_2, \G_3)$ with approximation ratio $\alpha$ means the following. In the search version: given a graph $G$ that contains a cut of weight $\rho$, find a triangle-free subgraph of weight $\alpha \rho$. In the decision version: given a graph $G$ and a number $\rho \in \N$, output \textsc{Yes} if it has a cut of weight $\rho$, and \textsc{No} if it has no triangle-free subgraph of weight $\alpha \rho$.

We define the problem $\PCSP(G, H)$ identically to $\maxPCSP(G, H)$, except that it is guaranteed that $\rho$ is the number of edges of $G$. Thus observe that $\PCSP(G, H)$ reduces to $\maxPCSP(G, H)$ trivially, in the sense that there is a polynomial-time reduction from $\PCSP(G,H)$ to $\maxPCSP(G,H)$ that does not change the input.
 
Suppose $G \to G' \to H' \to H$. Then, it follows that 
 $\PCSP(G, H)$ polynomial-time reduces to $\PCSP(G', H')$
and $\maxPCSP(G, H)$ polynomial-time reduces to $\maxPCSP(G', H')$ (and the same holds for $\alpha$-approximation).
Furthermore, the decision version of $\PCSP(G,H)$ and $\maxPCSP(G, H)$ polynomial-time reduces to the search version of $\PCSP(G,H)$ and $\maxPCSP(G,H)$, respectively. In other words, the decision version is no harder than the search version. Hence by proving our tractability results for the search version, and our hardness results for the decision version, we prove them for both versions of the problems.

\paragraph{SDP} For the Max-Cut problem, which is just $\maxPCSP(K_2,K_2)$, the
\emph{basic SDP relaxation} for a graph $G = (V, E)$ with $n$ vertices,
which can be solved within additive error $\epsilon$ in polynomial time with
respect to the size of $G$ and $\log(1 / \epsilon)$,\footnote{Throughout we will
assume, as is common in the literature, that numbers are represented with
arbitrary precision.}
is as follows:
\begin{maxi}|s|
{}{\sum_{(u, v) \in E} \frac{1 - \vx_u \cdot \vx_v}{2}} 
{}{}\label{eq:sdp}
  \addConstraint{\norm{\vx_u}^2 = 1}{} 
\addConstraint{\vx_u \in \R^n}. 
\end{maxi}

Goemans and Williamson~\cite{GW95} gave a randomised rounding algorithm for the SDP~\eqref{eq:sdp} with approximation ratio
\[
\alpha_{GW} = \left(\max_{0\leq \tau \leq \pi} \frac{\pi}{2}\frac{1 - \cos \tau}{\tau}\right)^{-1} = 0.878 \cdots,
\]
thus beating the trivial approximation ratio of $1/2$ obtained by, e.g., a random cut. Their algorithm solves the SDP~\eqref{eq:sdp}, selects a uniformly random hyperplane in $\R^n$, and returns the cut induced by the hyperplane.

\section{Results}\label{sec:results}

Our main result is the following theorem (already advertised in the Introduction).

\UGCdichotomy*

This is an optimisation variant of a conjecture by Brakensiek and Guruswami on
the tractability boundary of promise CSPs on undirected graphs.

\begin{conjecture}[\cite{BG21:sicomp}]\label{conjBG}
  Let $G$ and $H$ be two fixed graphs such that there is a homomorphism from $G$ to $H$. 
  If $G$ is bipartite then $\PCSP(G,H)$ is tractable. Otherwise, $\PCSP(G,H)$ is \NP-hard.
\end{conjecture}

The currently known cases supporting~\Cref{conjBG} are
\NP-hardness of $\PCSP(K_3,K_5)$~\cite{BBKO21}, $\PCSP(K_k,K_{\B{k}}-1)$ for
$k\geq 4$~\cite{KOWZ23}, and $\PCSP(C_{2k+1},K_4)$ for $k\geq 1$~\cite{Avvakumo25:stoc}, where $C_{2k+1}$ denotes a cycle on $2k+1$ vertices.
As $\PCSP(G, H)$ reduces  to $\maxPCSP(G, H)$, \Cref{conjBG} implies that
$\maxPCSP(G, H)$ is \NP-hard whenever $G$ is non-bipartite. We establish this
result under the $\UGC$ but \emph{not} relying on~\Cref{conjBG}. It follows from
our~\Cref{thm:UGCdichotomy} that not all cases of $\maxPCSP(G,H)$ with bipartite
$G$ are 1-approximable, and thus the tractability boundary lies elsewhere for 1-approximation.

We note that establishing~\Cref{thm:UGCdichotomy} appears easier than
resolving~\Cref{conjBG}, similarly to how the complexity of exact solvability of 
$\maxCSP$s~\cite{TZ16:jacm} was resolved before the complexity of decision
$\CSP$s~\cite{Bulatov17:focs,Zhuk20:jacm}.

\medskip
An important part in proving~\Cref{thm:UGCdichotomy} is the \NP-hardness of finding a triangle-free subgraph.
For this problem,  we also establish a non-trivial approximation result.

\begin{restatable}{theorem}{main}\label{thm:main}
$\maxPCSP(K_2, \G_3)$ is $0.8823$-approximable (in the search version) in polynomial time, and it is \NP-hard to $(25 / 26+\epsilon)$-approximate (even in the decision version) for any fixed $\epsilon > 0$.
\end{restatable}

Note that $0.8823 > 0.878\ldots$, thus our algorithm beats the Goemans-Williamson algorithm.
\Cref{thm:main} is proved in two parts: the tractability side
in~\Cref{sec:tractability} and the hardness side in~\Cref{sec:hardness}. We quickly give an intuitive explanation of why such an algorithm is possible.
Define
\[
\tau_{GW} = 
\arg \max_{0\leq \tau \leq \pi} \frac{\pi}{2}\frac{1 - \cos \tau}{\tau} \approx 0.742 \pi.
\]
The function $\tau \mapsto (\pi / 2) (1 - \cos \tau) / \tau$, depicted in~\Cref{fig:GW}, is increasing up to $\tau_{GW}$, then decreasing.
\begin{figure}[thbp]
    \centering
    \begin{tikzpicture}[scale=0.8,
    declare function={
    mycos(\t) = cos(180 * \t / pi);
    }
    ]
    \begin{axis}[
        xmin=0,xmax=pi,
        ymin=0,ymax=1.2,
        axis x line=middle,
        axis y line=middle,
        axis line style=<->,
        xlabel={$\tau$},
        ylabel={$\alpha$},
        extra x ticks = {2.3311,pi},
        extra x tick labels = {$\tau_{GW}$, $\pi$},
        extra y ticks = {1/0.878},
        extra y tick labels = {$\alpha_{GW}^{-1}$},
        grid=both,
        width=\textwidth*0.7,
           legend style={at={(0.65,.8)},anchor=north west},
        ]
        \addplot[no marks,myBrightBlue,thick] expression[domain=0:pi,samples=300]{(pi / 2) * (1 - mycos(x)) / x};
        \addlegendentry{$\tau \mapsto \frac{\pi}{2} \frac{1 - \cos \tau}{\tau}$}
        
    \end{axis}
    \end{tikzpicture}
    \caption{Function giving rise to $\alpha_{GW}, \tau_{GW}$.}
    \label{fig:GW}
\end{figure}
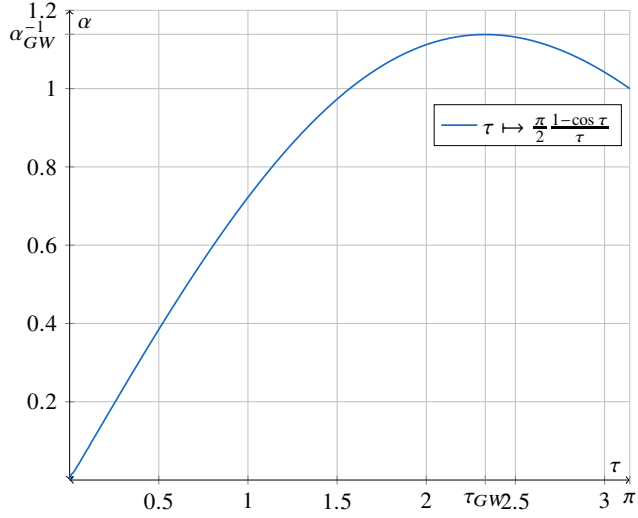
Why would $\maxPCSP(K_2, \G_3)$ be easier to approximate than $\maxPCSP(K_2, K_2)$? Consider the Goemans-Williamson algorithm for $\maxPCSP(K_2, K_2)$; the worst-case performance of this algorithm appears in a graph where the embedding into $\R^n$ given by solving the SDP \eqref{eq:sdp} gives all edges an angle of approximately $\tau_{GW}$. Observe that $\tau_{GW} > 2\pi /3$ --- so immediately (cf.~\Cref{lem:noTriangles}) this instance is triangle-free. So, for this instance an algorithm for $\maxPCSP(K_2, \G_3)$ could just return the entire graph! Indeed, in order to create an instance that contains triangles
one needs to introduce shorter edges. This suggests that a hybrid algorithm, that either selects ``long edges'' or some appropriate selection of ``shorter edges''  (still longer than some threshold), should have better performance.
The details can be found in~\Cref{sec:tractability}, while the \NP-hardness part
of~\Cref{thm:main} is proved in~\Cref{sec:hardness}.

\medskip
Our next result is a combination of
the Goemans-Williamson SDP for Max-Cut~\cite{GW95} and a rounding scheme due to Frieze
and Jerrum for Max-3-Cut~\cite{FJ97}.

\begin{restatable}{theorem}{thmktwokthree}\label{thm:k2k3}
$\maxPCSP(K_2, K_3)$ is 1-approximable in polynomial time (in the search version).
\end{restatable}

This algorithm is somewhat similar to the Goemans-Williamson algorithm, except
that rather than selecting a uniformly random hyperplane, it selects three
normally distributed vectors, and partitions the vertices according to which
vector they are closest to (where closeness is measured in terms of inner
products). Thus, while this algorithm solves the same SDP as the algorithm
from~\Cref{thm:main}, it rounds the solution differently. This rounding scheme is the same as the rounding scheme of~\cite{FJ97}; it is also very similar to one of the rounding schemes of~\cite{Karger98:jacm}. The details can be
found in~\Cref{sec:23}.

\medskip

The key in proving \NP-hardness in~\Cref{thm:UGCdichotomy} is the following
result, established by generalising the proof of Dinur, Mossel, and
Regev\cite{Dinur09:sicomp} with a multilayered unique games conjecture, in the
style of~\cite{BWZ21}. The details can be found in~\Cref{sec:ugchardness}.

\begin{restatable}{theorem}{thmUGChardness} \label{thm:ugcHardness}
    For every $k\geq 1$ and $\ell \geq 3$, 1-approximation of $\maxPCSP(C_{2k+1},
    K_\ell$) is \NP-hard assuming the \UGC.
\end{restatable}

We now have all tools to prove~\Cref{thm:UGCdichotomy}.

\begin{proof}[Proof of~\Cref{thm:UGCdichotomy}]
    If $G$ is bipartite and $H$ contains a triangle then $G \to K_2 \to K_3 \to
    H$, so it follows that $\maxPCSP(G, H)$  reduces to $\maxPCSP(K_2, K_3)$,
    which is 1-approximable by~\Cref{thm:k2k3}. 
    
    If $G$ is bipartite and $H$ is triangle-free then $K_2 \to G \to H \to
    \G_3$, so $\maxPCSP(K_2, \G_3)$ reduces to $\maxPCSP(G, H)$. Hence,
    by~\Cref{thm:main}, it is \NP-hard even to approximate $\maxPCSP(G, H)$ with
    approximation ratio $25 / 26 + \epsilon$ (in the decision version), and thus
    in particular 1-approximation of $\maxPCSP(G, H)$ is \NP-hard (in the decision version).

    If $G$ is not bipartite then we have $C_{2k + 1} \to G$ for some $k$, and,
    since $H$ is loopless, we have $H \to K_\ell$ for some $\ell \geq 3$. Hence
    \NP-hardness of 1-approximation under the \UGC{} follows
    from~\Cref{thm:ugcHardness}.
\end{proof}

We note that our hardness result depends in an essential way on the fact that the input graph can have multiple edges; we equivalently could have allowed non-negative integer weights on the edges. This variant of the problem is most natural when looking at it as a constraint satisfaction problem. It is  interesting to ask what the complexity of $\maxPCSP(K_2, \G_3)$ is if the input graph is  both weightless and without multiple edges.

\section{\texorpdfstring{Approximation of $\maxPCSP(K_2, {\mathfrak{G}}_3)$}{Approximation of maxPCSP(K2,G3)}}
\label{sec:tractability}

In this section, we will prove the tractability part of the following result.

\main*

We will need the following technical lemma.

\begin{lemma}\label{lem:bounds}
    There exist $\alpha \geq 0.88232, P, Q, \tau \in \R$ with $P + Q = 1, P \geq 0, Q \geq 0$, $\tau \in [2 \pi / 3, \tau_{GW}]$, such that the following hold
\begin{align}
    P \frac{\theta}{\pi} + Q \geq \alpha \frac{1 - \cos\theta}{2} && \theta \in [\tau, \pi],\label{eq:cons1}\\
    \frac{\phi}{\pi}\ =\ P \frac{\phi}{\pi} + Q \frac{\phi}{\pi} \geq \alpha
    \frac{1 - \cos\phi}{2} && \phi \in [\pi - \tau/2, \tau],\label{eq:cons2}\\
    P \frac{\psi}{\pi} \geq \alpha \frac{1 - \cos \psi}{2} && \psi \in [0, \pi- \tau/2].\label{eq:cons3}
\end{align}
In particular, we can take $\tau = 2.18746$, $Q = 1 - P$ and
\[
P = \frac{\alpha \pi}{2} \left(\frac{1 - \cos(\pi - \tau / 2)}{\pi - \tau / 2} \right) \approx 0.987535.
\]
\end{lemma}

The existence of these constants will be of key importance to our approximation algorithm; in particular, $\alpha$ will turn out to be the approximation ratio of our algorithm; we therefore are trying to make $\alpha$ as large as possible.

\begin{proof}
    Observe that the function $(1 - \cos x) / x$ increases monotonically for $0
    \leq x \leq \tau_{GW}$, and decreases monotonically for $\tau_{GW} \leq x
    \leq \pi$ (cf.~\Cref{fig:GW}). Note that the bound~\eqref{eq:cons3} is
    equivalently a bound of the form
    \[
    \frac{2P}{\pi \alpha} \geq \frac{1 - \cos \psi}{\psi}
    \]
    for $\psi \in [0, \pi - \tau / 2]$. Noting that $\tau \geq 2\pi / 3$, and hence $\pi - \tau / 2 \leq 2\pi / 3 \leq \tau_{GW}$, by monotonicity this bound holds for all $\psi \in [0, \pi - \tau / 2]$ if and only if it holds for $\psi = \pi - \tau / 2$, i.e.~if and only if
    \[
    P\geq \frac{\alpha\pi}{2} \underbrace{\left(\frac{1 - \cos (\pi - \tau/2)}{\pi - \tau/2} \right)}_{X(\tau)}.
    \]

    For~\eqref{eq:cons2}, since we will take $\tau
    \leq \tau_{GW}$, by a similar monotonicity argument~\eqref{eq:cons2} holds for all $\phi \in [\pi - \tau / 2, \tau]$ if and only if it holds for $\phi = \tau$, i.e.
    \[
    \alpha \leq \frac{2}{\pi} \frac{\tau}{1 - \cos \tau}.
    \]

    For \eqref{eq:cons1} we see that it becomes easier to satisfy if $P$ is smaller and $Q$ is larger, so we will choose $P$ to be the minimum value it could have vis-à-vis \eqref{eq:cons2} and \eqref{eq:cons3}, namely
    \[
    P = \frac{\alpha \pi}{2} X(\tau).
    \]
    Hence, as $P + Q = 1$, the first constraint becomes
    \[\frac{\alpha\theta}{2}  X(\tau)
     + 1 -  \frac{\alpha\pi}{2} X(\tau)
    \geq \alpha \frac{1 - \cos \theta}{2}.
    \]
    Simplifying, we get
    \[\alpha \frac{\theta - \pi }{2}X(\tau)
     + 1 
      \geq \alpha \frac{1 - \cos \theta}{2}.
    \]
    This is equivalent to
    \[
    \alpha \leq \frac{2}{(\pi - \theta) X(\tau) - \cos \theta + 1}.
    \]
    Separating out $\alpha$, we have
    \[
    \alpha \leq \min_{\theta} \frac{2}{(\pi - \theta) X(\tau) - \cos \theta + 1}
    =
    2 \left(\max_{\theta}(\pi - \theta) X(\tau) - \cos \theta + 1 \right)^{-1}.
    \]

    Now we can observe that $0.75 > X(\tau) > 0.6$ for $2\pi / 3 \leq \tau \leq \pi$ numerically. The derivative of the function we are maximising is $- X(\tau) + \sin \theta$ --- this is necessarily positive at $\theta = 2 \pi / 3$ and negative for $\theta = \pi$, because of the approximation of $X(\tau)$ above. So the maximum is hit at $\pi - \arcsin X(\tau)$, as this is the unique point where this derivative is 0 within the bound for $\theta$. Thus the bound is
    \[
    \alpha \leq \frac{2}{X(\tau) \arcsin X(\tau)  + \cos \arcsin X(\tau) + 1}.
    \]
    So we now want to find $\tau$ that maximises $\alpha$ such that
    \begin{align}
    \alpha & \leq \frac{2}{X(\tau) \arcsin X(\tau)  + \cos \arcsin X(\tau) + 1}, \label{eq:bound1}\\
    \alpha & \leq \frac{2}{\pi} \frac{\tau}{1 - \cos \tau}. \label{eq:bound2}
    \end{align}

    We can see this situation in~\Cref{fig:bounds}. Numerically we compute that if we choose $\tau = 2.18746$, then we get $\alpha \geq 0.88232$.
    \begin{figure}
    \centering
    \begin{tikzpicture}[
    declare function={
    mycos(\t) = cos(180 * \t / pi);
    X(\t) =(1 - mycos(pi - \t / 2)) / (pi - \t / 2);
    f2(\t) = (2 / pi) * \t / (1 - mycos(\t));
    f1(\t) = 2 / (rad(asin(X(\t))) * X(\t) + mycos(rad(asin(X(\t)))) + 1);},
    ]
    \begin{axis}[
        xmin=2.0944,xmax=2.331122375858596,
        ymin=0.870,ymax=0.9,
        axis x line=middle,
        axis y line=middle,
        axis line style=<->,
        y tick label style={/pgf/number format/.cd,fixed,fixed zerofill,precision=4},
        xlabel={$\tau$},
        ylabel={$\alpha$},
        extra x ticks = {2.18746},
        extra y ticks = {0.8823198361},
        xtick distance=0.06,
        grid=both,
        width=\textwidth*0.7
        ]
        \addplot[name path=f1,no marks,myBrightBlue,thick, forget plot] expression[domain=2.0944:2.331122375858596,samples=300]{f1(x)};
        \addplot[name path = f2,no marks,myBrightGreen,thick, forget plot] expression[domain=2.0944:2.331122375858596,samples=300]{f2(x)};

        \path[name path=axis] (axis cs:2.0944,0) -- (axis cs:2.331122375858,0);

       \addplot [
            pattern color=myBrightBlue,
            pattern=north east lines,
            area legend
        ]
        fill between[of=f1 and axis]; 
        \addlegendentry{Bound \eqref{eq:bound1}}
        
       \addplot [
            pattern color=myBrightGreen,
            pattern=north west lines,
            area legend
        ]
        fill between[of=f2 and axis]; 
        \addlegendentry{Bound \eqref{eq:bound2}}

        \addplot [black, mark = *, nodes near coords={Optimum},every node near coord/.style={anchor=270}] coordinates {(2.18746,0.8823198361)};
    \end{axis}
\end{tikzpicture}
    \caption{Bounds from~\Cref{lem:bounds}}
    \label{fig:bounds}
\end{figure}
\end{proof}

We now prove the desired tractability result.%
\footnote{We remark in passing that \emph{no} SDP-based algorithm can have performance greater than $8/9 = 0.888\ldots$, since for the triangle $K_3$ the SDP value is $9/4$, yet the largest triangle-free subgraph has weight $2$ (and $2 / (9/4) = 8/9$).}

\begin{theorem}\label{thm:approx}
    $\maxPCSP(K_2, \G_3)$ can be $0.8823$-approximated in polynomial time.
\end{theorem}
\begin{proof}
    Take $\alpha, P, Q, \tau$ as in~\Cref{lem:bounds}. Consider the following randomised algorithm.
    \begin{algoEnum}
        \item Input: a graph $G = (V, E)$, which admits an (unknown) maximum cut of weight $\rho$.
        \item Solving SDP \eqref{eq:sdp} to within error $\epsilon$, we get a set of vectors $\vx_v$ for $v \in V$, with $\norm{\vx_v}^2 = 1$ and 
        \[
        \sum_{(u, v) \in E} \frac{1 - \vx_u \cdot \vx_v}{2} \geq \rho - \epsilon.
        \]
        \item Flip a biased coin, randomly choosing from the following two cases.
        \begin{algoEnum}
        \item\label{case:cut} With probability $P$, sample a uniformly random hyperplane $H$, and compute the set of edges $(u, v)$ with $\vx_u$ on the opposite side of $H$ as $\vx_v$. Return this set of edges.
        \item\label{case:other} With probability $Q$, return all the edges $(u, v)$ with $\angle(\vx_u, \vx_v) > \tau$, then sample a uniformly random hyperplane $H$ and additionally return all the edges $(u, v)$ with $\angle(\vx_u, \vx_v) \geq \pi - \tau / 2$ and with $\vx_u, \vx_v$ on opposite sides of $H$.
        \end{algoEnum}
    \end{algoEnum}
    First, let us verify that our algorithm returns a triangle-free subgraph. First, in~\Cref{case:cut}, we return a bipartite subgraph, so we certainly return a triangle-free subgraph. The reasoning for~\Cref{case:other} is more geometric. Consider any three edges returned in this case. If the three edges are of angle between $\pi - \tau / 2$ and $\tau$, then they cannot form a triangle, since the edges of this kind that we return form a bipartite graph (determined by $H$). Conversely, suppose that at least one edge has angle greater than $\tau$, and the other two have angle at least $\pi - \tau / 2$. Then the sum of the angles of the edges are $> 2\pi$, and hence by~\Cref{lem:noTriangles}, they cannot form a triangle.

    Now, we compute the expected performance of our algorithm. Recall that for a uniformly random hyperplane $H$, two unit vectors $\vx, \vy$ are on opposite sides of $H$ with probability $\angle(\vx, \vy) / \pi$~\cite{GW95}.
    Consider any edge $(u, v)$ in our graph. We will show that the probability that the edge is included in the cut is at least $\alpha (1 - \vx_u \cdot \vx_v) / 2 = \alpha (1 - \cos \angle(\vx_u, \vx_v))/ 2$; there are three cases depending on $\angle(\vx_u, \vx_v)$.

    \begin{description}
    \item[$\bm{\angle(\vx_u, \vx_v) \in (\tau, \pi]}$.]
    In this case, the edge is included with probability $\angle(\vx_u, \vx_v) / \pi$ in~\Cref{case:cut}, and with probability 1 in~\Cref{case:other}. Thus, and by applying \eqref{eq:cons1} from~\Cref{lem:bounds}, we find that the edge is included with probability
    \[
    P \frac{\angle(\vx_u, \vx_v)}{\pi} + Q \geq \alpha\frac{1 - \cos \angle(\vx_u, \vx_v)}{2}.
    \]
    
    \item[$\bm{\angle(\vx_u, \vx_v) \in [\pi - \tau / 2, \tau]}$.]
    In this case, the edge is included with probability $\angle(\vx_u, \vx_v) / \pi$ both in~\Cref{case:cut} and in~\Cref{case:other}. Thus, by \eqref{eq:cons2} from~\Cref{lem:bounds}, the edge is included with probability
    \[
    P \frac{\angle(\vx_u, \vx_v)}{\pi} + Q
    \frac{\angle(\vx_u, \vx_v)}{\pi} \geq \alpha\frac{1 - \cos \angle(\vx_u, \vx_v)}{2}.
    \]
    
    \item[$\bm{\angle(\vx_u, \vx_v) \in [0, \pi - \tau / 2)}$.]
    In this case, the edge is included with probability $\angle(\vx_u, \vx_v) / \pi$ in~\Cref{case:cut} and never included in~\Cref{case:other}. Thus, by~\eqref{eq:cons3} from~\Cref{lem:bounds}, the edge is included with probability
    \[
    P \frac{\angle(\vx_u, \vx_v)}{\pi} \geq \alpha\frac{1 - \cos \angle(\vx_u, \vx_v)}{2}.
    \]
    \end{description}

    Hence overall we include an edge $(u, v)$ with probability at least $\alpha (1 - \vx_u \cdot \vx_v) / 2$. Since
    \[
    \sum_{(u, v) \in E} \frac{1 - \vx_u \cdot \vx_v}{2} \geq \rho - \epsilon,
    \]
    it follows that we return a triangle-free subgraph with expected weight $\alpha (\rho - \epsilon)$.

    Now, we derandomise our algorithm. First, rather than randomly choosing between~\Cref{case:cut} and~\Cref{case:other}, we can just run both cases and return the better of the two solutions. Second, using the techniques of~\cite{Mahajan99:sicomp} (or, more efficiently~\cite{BK05}), we can derandomise the choice of random hyperplane in polynomial time, at the cost of losing $\epsilon'$ potential value, in polynomial time in $\log(1 / \epsilon')$. Hence, in polynomial time in the size of the graph and $\epsilon, \epsilon'$ we return a triangle-free subgraph with weight $\alpha (\rho - \epsilon) - \epsilon' \geq \alpha \rho - (\epsilon + \epsilon')$.

    To complete the proof, note that we take $\alpha \geq 0.88232$, but we only advertise an approximation ratio of $0.8823$. Recalling that any graph admits a cut with at least half its edges, we see that $\rho \geq |E| / 2$. Hence, the subgraph we return has weight $0.8823 \rho + 2 \cdot 10^{-5} \rho - (\epsilon + \epsilon')$, and since $\rho \geq |E|/2 \geq 1/2$ (since if $|E| = 0$ the problem is trivial), this is at least $0.8823 \rho + (10^{-5} - \epsilon - \epsilon')$. Hence it is sufficient to choose $\epsilon + \epsilon' = 10^{-5}$. The algorithm then runs in polynomial time in the size of $G$, and finds a triangle-free subgraph with weight $0.8823 \rho$, as required.
\end{proof}

It is also interesting to consider what the power of our approximation algorithm
is in the \emph{almost satisfiable regime}, i.e.~if an input graph that has a
cut of value $1 - \epsilon$. It turns out that in this case we output a
triangle-free subgraph with $1 - O(\epsilon)$ edges, significantly more than the
$1 - O(\sqrt{\epsilon})$ edges outputted by the Goemans-Williamson
algorithm~\cite{GW95}.
This is not very hard to see, it follows immediately from the fact that our algorithm can choose all edges of angle $> \tau$ immediately.

\begin{theorem}\label{thm:almost}
    The derandomised algorithm from \Cref{thm:approx}, if run on an input graph $G$ with a cut with a  $(1 - \epsilon)$-fraction of edges, produces a triangle-free subgraph with $(1 - O(\epsilon))$-fraction of edges.
\end{theorem}
Indeed, the (extremely loose) analysis below gives us that it returns a triangle-free subgraph with a $(1 - 15\epsilon)$-fraction of edges at least.
\begin{proof}
    Suppose we solve the SDP \eqref{eq:sdp} within error $\epsilon$, hence getting a solution with value at least $(1 - \epsilon)^2 \geq 1 - 3\epsilon$ times the number of edges in $G$. Suppose that of the edges, a proportion of $a$ of them have angle $> \tau$, and a proportion of $b$ of them have angle $\leq \tau$. We will output at least all of the edges with angle $> \tau$, so we must show that $a = 1 - O(\epsilon)$. Observing that in the worst case all edges counting towards $a$ have angle $\pi$ and all edges counting towards $b$ have angle $\tau$, we have
    \begin{align*}
    a + b &= 1 \\
    a + b\underbrace{\frac{1 - \cos(\tau)}{2}}_{C = C(\tau)} &\geq 1 - 3\epsilon.
    \end{align*}
    Note that $\tau$ is an absolute constant, and thus $C$ is also an absolute constant. Indeed we can approximate $C \approx 0.789$. Subtracting $C$ times the first equation from the second we get $ (1 - C)a = (1 - C) - 3\epsilon$, and then dividing again by $1 - C$ to get $a = 1 - (3 / (1 - C)) \epsilon \geq 1 - 15 \epsilon = 1 - O(\epsilon)$, as required.
\end{proof}

Interestingly, the non-derandomised algorithm has worse performance! Since it chooses at random between selecting all long edges deterministically and cutting according to the Goemans-Williamson algorithm, the performance degrades to $1 - O(\sqrt{\epsilon})$.

\section{\texorpdfstring{Approximation of $\maxPCSP(K_2,K_3)$}{Approximation of maxPCSP(K2,K3)}}
\label{sec:23}

We first introduce some useful notation.
For any predicate $\phi$, we let $[\phi] = 1$ if $\phi$ is true, and $0$ otherwise.

For an event $\phi$ we let $\Pr [\phi]$ be the probability that $\phi$ is true. For a random variable $X$, we let $\mathbb{E}[X]$ denote its expected value. Note that $\mathbb{E}[ [\phi]] = \Pr[\phi]$.
For any two distributions $\mathcal{D}, \mathcal{D}'$ with domains $A, A'$, we let $\mathcal{D} \times \mathcal{D}'$ denote the product distribution, whose domain is $A \times A'$.
For any distribution $\mathcal{D}$ over $\R$ and $a, b \in \R$, the distribution $a\mathcal{D} + b$ is the distribution of $aX + b$ when $X \sim \mathcal{D}$.
We use the standard probability theory abbreviations i.i.d.~(independent and identically distributed) and p.m.f.~(probability mass function).

We introduce a few classic distributions we will need. The uniform distribution $\mathcal{U}(D)$ over a discrete set $D$ is the distribution with p.m.f.~$f : D \to [0, 1]$ given by $f(x) = 1 / |D|$. Note that $\mathcal{U}(D^n)$ is the same as ${\mathcal{U}(D)}^n$, a fact which we will use implicitly. We let $\NBin(n)$ denote a normalised binomial distribution: it is the distribution of $X_1 + \cdots + X_n$, where $X_i \sim \mathcal{U}(\{-1/\sqrt{n}, 1/\sqrt{n}\})$. The domain of this distribution is $\{(-n + 2k) / \sqrt{n} \mid 0 \leq k \leq n\}$, the probability mass function is $(-n + 2k)/\sqrt{n} \mapsto \binom{n}{k} / 2^n$, the expectation is 0, and the variance is 1.
If $\mu, \sigma \in \R$, then we let $\mathcal{N}(\mu, \sigma^2)$ denote the
normal distribution with mean $\mu$ and variance $\sigma^2$. Fixing $d$, if
$\mathbf{\mu} \in \R^d, \mathbf{\Sigma} \in \R^{d\times d}$, then we let
$\mathcal{N}(\mathbf{\mu}, \mathbf{\Sigma})$ denote the multivariate normal
distribution with mean $\mathbf{\mu}$ and covariance matrix $\mathbf{\Sigma}$.
(Recall that the covariance matrix is a square matrix containing the covariances
between all pairs of variables.) 
We let $\mathbf{I}_d$ denote the $d \times d$ identity matrix. Observe that if $\mathbf{x} \sim \mathcal{N}( \mathbf{\mu}, \mathbf{\Sigma})$, where $\mathbf{x} \in \R^d$, then for any matrix $A \in \R^{d' \times d}$ we have that $A \mathbf{x} \sim \mathcal{N}(A \mathbf{\mu}, A \mathbf{\Sigma} A^T)$. Furthermore if $\mathbf{x} \sim \mathcal{N}(\mathbf{\mu}, \mathbf{\Sigma})$ with $\mathbf{\Sigma}$ positive semidefinite, then by finding the Cholesky decomposition $\mathbf{\Sigma} = \mathbf{A} \mathbf{A}^T$, where $\mathbf{A} \in \R^{d \times d}$, we find that $\vx$ is identically distributed to $\mathbf{A} \vx' + \mathbf{\mu}$, where $\vx' \sim \mathcal{N}(\mathbf{0}, \mathbf{I}_d)$.

Our goal is to prove the following result.

\thmktwokthree*

Our proof will be split into three parts: First we prove some technical bounds which we will need. Next, we provide a randomised algorithm. Finally, we derandomise the algorithm.

\paragraph{Technical bounds.} For the proof, we will need a technical lemma,
stated as~\Cref{lem:prExpr} below. The proof of~\Cref{lem:prExpr} is an application of the following result of Cheng.

\begin{theorem}[{\cite{Cheng68}\cite[Equation (2.18)]{Cheng69}}]\label{thm:cheng}
Suppose $\mathbf{u} = (u_1, u_2, u_3, u_4) \sim \mathcal{N}(\mathbf{0}, \mathbf{\Sigma})$ are drawn from a quadrivariate normal distribution with mean zero and covariance matrix
\[
\mathbf{\Sigma} =
\begin{pmatrix}
   1 & a & b & ab \\
   a & 1 & ab & b \\
   b & ab & 1 & a \\
   ab & b & a & 1 \\
\end{pmatrix},
\]
where $a, b \in [-1, 1]$. Then $\Pr_{\mathbf{u}}[u_1 \geq 0, u_2 \geq 0, u_3 \geq 0, u_4 \geq 0]$ is 
\[
\frac{1}{16} + \frac{\arcsin a + \arcsin b + \arcsin ab}{4\pi} + \frac{{(\arcsin a)}^2 + {(\arcsin b)}^2 - {(\arcsin ab)}^2}{4\pi^2}.
\]
\end{theorem}

\begin{lemma}\label{lem:prExpr}
Fix $\alpha, \beta \in \mathbb{R}$ such that $\alpha^2 + \beta^2 = 1$. Suppose $x_1, x_2, x_3, y_1, y_2, y_3 \sim \mathcal{N}(0, 1)$ i.e.~they are i.i.d.~standard normal variables. The probability that
\begin{align*}
x_1 & \geq x_2 \\
x_1 & \geq x_3 \\
\alpha x_1 + \beta y_1 & \geq  \alpha x_2 + \beta y_2 \\
\alpha x_1 + \beta y_1 & \geq  \alpha x_3 + \beta y_3 \\
\end{align*}
is precisely
\[
P(\alpha) = 
\frac{1}{9}
+ \frac{\arcsin \alpha + \arcsin \frac{\alpha}{2}}{4\pi} 
+ \frac{
{(\arcsin \alpha)}^2-
{(\arcsin \frac{\alpha}{2})}^2
}{4\pi^2}.
\]
\end{lemma}
\begin{proof}
Define the following normally distributed random variables:
\begin{align*}
u_1 &= (x_1 - x_2)/\sqrt{2}, \\
u_2 &= (x_1 - x_3)/\sqrt{2}, \\
u_3 &= (\alpha x_1 + \beta y_1 - \alpha x_2 - \beta y_2)/\sqrt{2}, \\
u_4 &= (\alpha x_1 + \beta y_1 - \alpha x_3 - \beta y_3)/\sqrt{2}.
\end{align*}
Equivalently, we define
\[
\begin{pmatrix}
    u_1 \\ u_2 \\ u_3 \\ u_4
\end{pmatrix}
=
\underbrace{
\frac{1}{\sqrt{2}}
\begin{pmatrix}
    1 & -1 & 0 & 0 & 0 & 0 \\
    1 & 0 & -1 & 0 & 0 & 0 \\
    \alpha & -\alpha & 0 & \beta & -\beta & 0 \\
    \alpha & 0 & -\alpha & \beta & 0 & -\beta 
\end{pmatrix}}_A
\begin{pmatrix} x_1 \\ x_2 \\ x_3 \\ y_1 \\ y_2 \\ y_3\end{pmatrix}.
\]
Hence the covariance matrix $\mathbf{\Sigma}$ of $(u_1, u_2, u_3, u_4)$ is given by $AA^T$
which (noting that since $\alpha^2 + \beta^2 = 1$) we compute to be
\[
\mathbf{\Sigma} = \begin{pmatrix}
1 & \frac{1}{2} & \alpha & \frac{\alpha}{2} \\
\frac{1}{2} & 1 & \frac{\alpha}{2}& \alpha \\
 \alpha & \frac{\alpha}{2} & 1 & \frac{1}{2} \\
\frac{\alpha}{2}& \alpha & \frac{1}{2} & 1 \\
\end{pmatrix}.
\]
In particular, $(u_1, u_2, u_3, u_4) \sim \mathcal{N}(\mathbf{0}, \mathbf{\Sigma})$.

The probability we want is just $\Pr_{\mathbf{u}}[ u_1 \geq 0, u_2 \geq 0, u_3 \geq 0, u_4 \geq 0]$. Apply~\Cref{thm:cheng} with $a = \frac{1}{2}$ and $b = \alpha$ to find that the required probability is
\begin{multline*}
\frac{1}{16} + \frac{\arcsin \frac{1}{2} + \arcsin \alpha + \arcsin \frac{\alpha}{2}}{4\pi} + \frac{{(\arcsin \frac{1}{2})}^2 + {(\arcsin \alpha)}^2 - {(\arcsin \frac{\alpha}{2})}^2}{4\pi^2} = \\
\frac{1}{9}
+ \frac{\arcsin \alpha + \arcsin \frac{\alpha}{2}}{4\pi} 
+ \frac{{(\arcsin \alpha)}^2-{(\arcsin \frac{\alpha}{2})}^2}{4\pi^2},
\end{multline*}
as required.
\end{proof}

We will need a bound on the $P(\alpha)$ function from~\Cref{lem:prExpr}.

\begin{lemma}\label{lem:ineq}
    For $-1 \leq \alpha \leq +1$, $1 - 3P(\alpha) \geq \frac{1 - \alpha}{2}$.
\end{lemma}
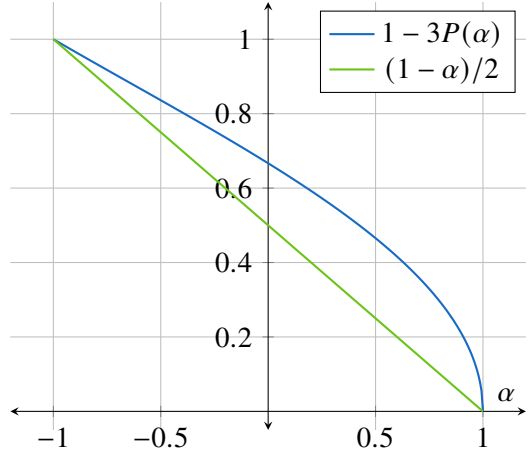
\begin{figure}
\begin{center}
\begin{tikzpicture}[>=stealth]
    \begin{axis}[
        xmin=-1.2,xmax=1.2,
        ymin=-0.05,ymax=1.1,
        axis x line=middle,
        axis y line=middle,
        axis line style=<->,
        xlabel={$\alpha$},
        grid=both,
        ]
        \addplot[name path=a, no marks,myBrightBlue,thick] expression[domain=-1:1,samples=300]{1 - 3*(1/9 + (rad(asin(x)) + rad(asin(x/2)))/(4*pi) + (rad(asin(x))^2 - rad(asin(x/2))^2)/(4*pi*pi))};
        \addlegendentry{$1 - 3P(\alpha)$}
        \addplot[name path=b, no marks,myBrightGreen,thick] expression[domain=-1:1,samples=100]{(1-x)/2};
        \addlegendentry{$(1 - \alpha)/2$}
    \end{axis}
\end{tikzpicture}
\end{center}
\caption{Plot of expressions from~\Cref{lem:ineq}.}\label{fig:plots}
\end{figure}
The functions involved are shown in~\Cref{fig:plots}.
\begin{proof}
    Define 
    \[
    f(\alpha) = 1 - 3P(\alpha) - \frac{1 - \alpha}{2} = \frac{1}{6} + \frac{\alpha}{2} - \frac{3}{4\pi}\left(\arcsin \alpha + \arcsin \frac{\alpha}{2}\right) - \frac{3}{4\pi^2}\left({\left(\arcsin \alpha\right)}^2 - {\left(\arcsin \frac{\alpha}{2}\right)}^2\right).
    \]
    We want to show that $f(x) \geq 0$ for $x \in [-1, -1]$. First we show that $f(x) \geq 0$ for $x \in [-1, 0]$. By numeric computation,\footnote{Why not simply verify that the original inequality holds numerically? That inequality has an equality case precisely at the edges of the domain on which we want it to hold. Checking such an inequality numerically is very prone to numerical errors around the endpoints. On the other hand, the fact that $f'''(x) < 0$ for $-1 < x < 1$ holds strictly, and with considerable slack of about $0.1$. Furthermore, the maximum of $f'''$ even falls somewhere in the middle of the interval $(-1, 1)$, and the inequality holds in the limit towards $-1$ and 1. In such cases, even if the minimum calculation is slightly off, we can still be confident that $f'''(x) < 0$ on all of $(-1, 1)$.}
    we can verify that
    \[
    \max_{-1 < x < 1} f'''(x) \approx -0.1454 < 0,
    \]
    at $x \approx-0.5681$. Thus $f''(x)$ is decreasing, and $f'(x)$ is concave. Thus, by Jensen's inequality, for $x \in (-1, 0)$,
    \[
    f'(x) \geq  -x\lim_{t \to -1^+}f'(t) + (1+x)f'(0),
    \]
    (as $f'$ is not defined at $-1$). But $f'(0) \approx 0.1419 > 0$, and $\lim_{t \to -1^+} f'(t) \approx 0.1642 > 0$, so $f'(x) > 0$ for $x \in (-1, 0)$. It follows that $f$ is increasing on $[-1, 0]$, which is sufficient to show that $f(x) \geq 0$ for $x \in [-1, 0]$, as $f(-1) = 0$.

    Now, we consider $x \in [0, 1]$. Observe again that we know that $f''(x)$ is decreasing. But since $f''(0) \approx -0.1139 < 0$, it follows that $f''(x) < 0$ for $x \in [0,1)$; so $f$ is concave on $[0, 1]$. Now, again applying Jensen's inequality, we find that for $x \in [0, 1]$,
    \[
    f(x) \geq xf(0) + (1 - x) f(1).
    \]
    Now, $f(0) = 1/6 > 0$, and $f(1) = 0$, so $f(x) \geq 0$ for $x \in [0, 1]$.

    Thus our conclusion follows in all cases.
\end{proof}

\paragraph{Randomised algorithm.} We first provide a randomised version of our algorithm, accurate up to some $\epsilon$.

\begin{theorem}[Randomised version of~\Cref{thm:k2k3}]
    There exists a randomised algorithm which, given a graph $G = (V, E)$ that has a cut with $\rho$ edges and an accuracy parameter $\epsilon$, finds a 3-colouring of $G$ that satisfies $\rho-\epsilon$ edges in expectation, in polynomial time with respect to the size of $G$ and $\log(1 / \epsilon)$.
\end{theorem}

\begin{proof}
Our algorithm is as follows. In essence, it solves the SDP of Goemans and Williamson~\cite{GW95} i.e.~\eqref{eq:sdp}, then randomly rounds as the 3-colouring algorithm of Frieze and Jerrum~\cite{FJ97}.

\begin{algoEnum}
    \item Input: a graph $G = (V, E)$, which admits a cut of weight $\rho$. Suppose $n = |V|$.
    \item Solve SDP \eqref{eq:sdp} to within error $\epsilon$. Thus we get a set of vectors $\vx_v \in \R^n$ for $v \in V$, with $\norm{\vx_v}^2 = 1$ and 
    \[
    \sum_{(u, v) \in E} \frac{1 - \vx_u \cdot \vx_v}{2} \geq \rho - \epsilon.
    \]
    \item Sample three i.i.d.~normally distributed $n$-dimensional vectors $\va_1, \va_2, \va_3 \sim \mathcal{N}(\mathbf{0}, \mathbf{I}_n)$.
    \item Set the colour of node $u$ to $\arg \max_i \vx_u \cdot \va_i$. (Break ties arbitrarily.)
\end{algoEnum}

Let us compute the expected number of edges which are satisfied by this 3-colouring. Consider an edge $(u, v) \in E$; in terms of $\frac{1 - \vx_u \cdot \vx_v}{2}$, what is the probability that $(u, v)$ is properly coloured? This is the same as the probability that $\arg \max_i \vx_u \cdot \va_i \neq \arg \max_i \vx_v \cdot \va_i$, which, by symmetry, is equal to
\begin{equation}\label{eq:prob}
1 - 3 \Pr_{\va_1, \va_2, \va_3}[
\vx_u \cdot \va_1 \geq \vx_u \cdot \va_2,
\vx_u \cdot \va_1 \geq \vx_u \cdot \va_3,
\vx_v \cdot \va_1 \geq \vx_v \cdot \va_2,
\vx_v \cdot \va_1 \geq \vx_v \cdot \va_3].
\end{equation}
Let $\alpha = \vx_u \cdot \vx_v, \beta = \sqrt{1 - \alpha^2}$. 
Let us now consider an orthonormal basis $(\ve_1, \ldots, \ve_n)$, where we take $\ve_1 = \vx_u$, and take $\ve_2$ such that $\vx_v = \alpha \ve_1 + \beta \ve_2$. (If $\vx_u = \vx_v$ then we can take any unit vector orthogonal to $\ve_1$; otherwise we take one of the two unit vectors orthogonal to $\ve_1$ in the plane spanned by $\vx_u, \vx_v$.) Since $\va_1, \va_2, \va_3$ are taken from the standard $n$-variate normal distribution, we see that their projections to each $\ve_i$ are i.i.d.~standard normal variables. Hence, by letting $\va_i \cdot \ve_1 = a_i$ and $\va_i \cdot \ve_2 = b_i$ for $i = 1, 2, 3$, we can rewrite~\eqref{eq:prob} as
\begin{equation}\label{ineq:final}
1 - 3 \Pr_{\substack{a_1, a_2, a_3 \\ b_1, b_2, b_3}}[a_1 \geq a_2, a_1 \geq a_3, \alpha a_1 + \beta b_1 \geq \alpha a_2 + \beta b_2, \alpha a_1 + \beta b_1 \geq \alpha a_3 + \beta b_3 ],
\end{equation}
where $a_1, a_2, a_3, b_1, b_2, b_3$ are i.i.d.~standard normal variables. 
By~\Cref{lem:prExpr}, the expression in \eqref{ineq:final} is $1 - 3P(\alpha)$, which is, by~\Cref{lem:ineq}, at least $\frac{1 - \alpha}{2} = \frac{1 - \vx_u \cdot \vx_v}{2}$. Thus, by linearity of expectation, the expected number of edges we satisfy is at least
$\sum_{(u, v) \in E}\frac{1 - \vx_u \cdot \vx_v}{2} \geq \rho - \varepsilon$, as required.
\end{proof}

\paragraph{Derandomised algorithm.}
We will now show how to derandomise our algorithm, using
the method of conditional expectations, which was also used by Mahajan
and Ramesh~\cite{Mahajan99:sicomp}. The approach of Bhargava and
Kosaraju~\cite{BK05} derandomises conditional probabilities by an approximation
of normal distributions via polynomials; we approximate by a
scaled binomial distribution. We believe that the results of the literature are sufficient to prove the derandomisation theorem we need (which crucially needs to work for our 1-approximation setting); however we propose a simpler derandomisation method that we believe will be easier to apply in general. (Indeed, our method avoids integration altogether.) There is an interesting duality between our approach and that of Mahajan and Ramesh: we discretise the normal distribution, whereas they discretise the SDP vectors.

Our goal will be the following general derandomisation theorem.

\begin{theorem}\label{thm:derandGaussian}
    Fix a constant $d$. There exists an algorithm that does the following. Suppose we are given $n, m \in \mathbb{N}$, $\vx_{ij} \in \mathbb{R}^n$ and $y_{ij}, \varepsilon \in \mathbb{R}$ for all $i \in [m], j \in [d]$. Suppose $\va = (a_1, \ldots, a_n) \sim \mathcal{N}(\mathbf{0}, \mathbf{I}_n)$ and that
    \[
    \sum_{i = 1}^m \Pr_{\va}\left[ \bigwedge_{j = 1}^d \vx_{ij} \cdot \va > y_{ij} \right] \geq \rho
    \]
    for some $\rho \in \mathbb{R}$. Then the algorithm computes some particular $\va^* = (a_1^*, \ldots, a_n^*) \in \mathbb{R}^n$ such that
    \[
    \sum_{i =1 }^m \left[ \bigwedge_{j = 1}^d \vx_{ij} \cdot \va^* > y_{ij} \right] \geq \rho - \varepsilon,
    \]
    in polynomial time with respect to $n, m, 1/\varepsilon$.
\end{theorem}

To facilitate the proof of~\Cref{thm:derandGaussian}, we will need a multidimensional version of the Berry-Esseen theorem. We will use the following version with explicit constants, due to Rai\v{c}~\cite{Raic19}.

\begin{theorem}[{\cite[Theorem~1.1]{Raic19}}]\label{thm:BE}
Suppose $\mathbf{t}_1, \ldots, \mathbf{t}_N \in \mathbb{R}^d$ are independent random variables with mean zero, such that the sum of their covariance matrices is $\mathbf{I}_d$. Let $\mathbf{s} = \mathbf{t}_1 + \cdots + \mathbf{t}_N$. Suppose $\va \sim \mathcal{N}(\mathbf{0}, \mathbf{I}_d)$, and let $C \subseteq \mathbb{R}^d$ be convex and measurable. Then
\[
| \Pr_{\mathbf{s}}[\mathbf{s} \in C] - \Pr_{\va}[\va \in C] | \leq \left(42\sqrt[4]{d} + 16\right) \sum_{i = 1}^N \mathbf{E} \left[ \norm{ \mathbf{t}_i }^3 \right].
\]
\end{theorem}

The following is an easy and well-known corollary of~\Cref{thm:BE}: We can approximate a multivariate normal distribution with binomial distributions. For completeness, we provide a proof.

\begin{corollary}\label{corr:appliedBerryEsseen}
Let $d \in \mathbb{N}$ be a constant and take $\varepsilon \in (0, 1)$. Take
\begin{equation}\label{eq:defN}
N = N_\varepsilon \geq {\left(\frac{42 d^{7/4} +16 d^{3/2} }{\varepsilon}\right)}^2 = \frac{\xi_d}{\varepsilon^2},
\end{equation}
where $\xi_d = O(d^{7/2})$ depends only on $d$. Suppose
$s_1, \ldots, s_d \sim \NBin(N)$ are i.i.d., and let $\mathbf{s} = (s_1, \ldots, s_d)$.
 Let $\va \sim \mathcal{N}(\mathbf{0}, \mathbf{I}_d)$. Then for all convex measurable sets $C \subseteq \mathbb{R}^d$ we have
\[
| \Pr_{\mathbf{s}}[ \mathbf{s} \in C] - \Pr_{\va}[\va \in C] | \leq \varepsilon.
\]
\end{corollary}
\begin{proof}
Note that each component of $\mathbf{s}$ is i.i.d.~and distributed as the sum of $N$ independent trials that take value $\pm 1 / \sqrt{N}$ equiprobably. 
In other words, we can see
$\mathbf{s}$ as the sum $\mathbf{t}_1 + \cdots  + \mathbf{t}_N$, where $\mathbf{t}_1, \ldots, \mathbf{t}_N \sim \mathcal{U}({\{-1/\sqrt{N}, 1/\sqrt{N}\}}^d)$ are i.i.d.
Observe that the covariance matrix of $\mathbf{t}_i$ is $\mathbf{I}_d / N$, so the sum of these covariance matrices for all $i$ is $\mathbf{I}_d$. Furthermore $\norm{ \mathbf{t}_i } = \sqrt{(\pm 1 / \sqrt{N})^2 + \cdots + (\pm 1 / \sqrt{N})^2 } = \sqrt{d / N}$ with probability 1.
Now, apply~\Cref{thm:BE} to $\mathbf{s} = \sum_{i = 1}^N \mathbf{t}_i$. We find that
    \[
    |\Pr_{\mathbf{s}}[\mathbf{s} \in C] - \Pr_\va[\va \in C]| \leq
    \left(42 \sqrt[4]{d} + 16\right)
    \sum_{i = 1}^N \mathbb{E}\left[ \left|\left| \mathbf{t}_i \right|\right|^3 \right]
    =
    \left(42 \sqrt[4]{d} + 16\right) N\sqrt{\frac{d}{N}}^3 = \frac{42 d^{7/4} + 16 d^{3/2}}{\sqrt{N}}.
    \]
    Substituting~\eqref{eq:defN}, it follows that $|\Pr_{\mathbf{s}}[ \mathbf{s} \in C] - \Pr_{\va}[\va \in C]| \leq \epsilon$, as required.
\end{proof}

\begin{theorem}\label{thm:stepApprox}
    Fix a constant $d$, and take $\vx_{1}, \ldots, \vx_d \in \mathbb{R}^n, y_1, \ldots, y_d \in \mathbb{R}, z_1, \ldots, z_d \in \mathbb{R}, \epsilon \in \mathbb{R}$. Consider the function
    \[
    p(t) = \Pr_{\va \sim \mathcal{N}(\mathbf{0}, \mathbf{I}_n)}\left[\bigwedge_{i=1}^d \vx_{i} \cdot \va + z_i t > y_i \right].
    \]
    There exists a step function $\widehat{p}$ with $\poly(1/\epsilon)$ steps, where the steps and the values at those steps are computable in polynomial time with respect to $1 / \epsilon$
    and $n$, such that $| \widehat{p}(t) - p(t)| \leq \epsilon$ for all $t \in \mathbb{R}$.
\end{theorem}
\begin{proof}
    Observe that the tuple $(\vx_1 \cdot \va, \ldots, \vx_d \cdot \va)$ (interpreted as a column vector) is a $d$-variate normally distributed vector; namely, if we let
    \[
    \mathbf{X} =
    \begin{pmatrix}
        \vx_1^T \\ 
        \vdots \\
        \vx_d^T
    \end{pmatrix}
    \]
    be a block matrix whose rows are $\vx_1^T, \ldots, \vx_d^T$, then $(\vx_1 \cdot \va, \ldots, \vx_d \cdot \va) = \mathbf{X} \va \sim \mathcal{N}(\mathbf{0}, \mathbf{X} \mathbf{I}_D \mathbf{X}^T) = \mathcal{N}(\mathbf{0}, \mathbf{X} \mathbf{X}^T)$.
    
    We can compute the covariance matrix, namely $\mathbf{X} \mathbf{X}^T$, in polynomial time with respect to $n$. Now, by computing the Cholesky decomposition of this positive semidefinite matrix, we can find $\mathbf{X}' \in \R^{d \times d}$ such that $\mathbf{X}' \mathbf{X}'^T = \mathbf{X}\mathbf{X}^T$. Thus $(\vx_1 \cdot \va, \ldots, \vx_d \cdot \va) \sim \mathcal{N}(\mathbf{0}, \mathbf{X}' \mathbf{X}'^T)$.
    Letting $\vx_1'^T, \ldots, \vx_d'^T$ be the rows of $\mathbf{X}'$, we find that $(\vx_1 \cdot \va, \ldots, \vx_d \cdot \va)$ is identically distributed to $(\vx'_1 \cdot \va', \ldots, \vx_d' \cdot \va')$, when $\va \sim \mathcal{N}(\mathbf{0}, \mathbf{I}_n)$ and $\va' \sim \mathcal{N}(\mathbf{0}, \mathbf{I}_d)$, since both follow the distribution $\mathcal{N}(0, \mathbf{X}' \mathbf{X}'^T)$. Thus,
    \[
    p(t) = \Pr_{\va' \sim \mathcal{N}(\mathbf{0}, \mathbf{I}_d)}\left[\bigwedge_{i=1}^d \vx_{i}' \cdot \va' + z_i t > y_i \right].
    \]
    In other words, we have reduced the dimensionality of our problem from $n$ to $d$, a constant.

    Note that the set defined by $\bigwedge_{i = 1}^d \vx_i' \cdot \va' + z_i t
    > y_i$ is necessarily convex and measurable, being the intersection of
    finitely many half-spaces. Thus we can apply~\Cref{corr:appliedBerryEsseen}. Let $N = N_\epsilon$ and $s_1, \ldots, s_N \sim \NBin(N)$, and suppose $\mathbf{s} = (s_1, \ldots, s_d)$. Then, we know that 
    \[
    \left| p(t) - \Pr_{\mathbf{s}}\left[ \bigwedge_{i=1}^d \vx_i' \cdot \mathbf{s} + z_i t > y_i \right] \right| \leq \epsilon.
    \]
    This suggests using the following definition:
    \[
    \widehat{p}(t) =  \Pr_{\mathbf{s}}\left[ \bigwedge_{i=1}^d \vx_i' \cdot \mathbf{s} + z_i t > y_i  \right],
    \]
    as this must satisfy the condition
    $| \widehat{p}(t) - p(t)| \leq \epsilon$. What remains is to show that $\widehat{p}$ is a step function, and that these steps can be efficiently computed.

    Intuitively, this is the case since the probability distribution we define $\widehat{p}$ over is discrete. More precisely, letting $D = {\{ -N/\sqrt{N}, (-N+2)/\sqrt{N}, \ldots, (N-2)/\sqrt{N}, N/\sqrt{N} \}}^d$ be the domain of $\mathbf{s}$, and letting $q$ be the p.m.f.~of $\mathbf{s}$ (note that it can be efficiently computed, since $\mathbf{s}$ is distributed according to a product distribution of normalised binomials), we note that
    \begin{multline*}
    \widehat{p}(t) =  \Pr_{\mathbf{s}}\left[ \bigwedge_{i=1}^d \vx_i' \cdot \mathbf{s}  + z_i t > y_i \right]
    =
    \sum_{\mathbf{u} \in D}  q(\mathbf{u}) \left[ \bigwedge_{i = 1}^d \vx_i' \cdot \mathbf{u} + z_i t > y_i \right]
    \\=
    \sum_{\mathbf{u} \in D} q(\mathbf{u}) \prod_{i = 1}^d \left[ \mathbf{x}_i' \cdot \mathbf{u} + z_i t > y_i \right].
    \end{multline*}
    Now observe that, for each $\mathbf{u} \in D$, the function
    \[
    t \mapsto \left[ \mathbf{x}_i' \cdot \mathbf{u} + z_i t > y_i \right]
    \]
    is a step function with at most one step: if $z_i = 0$ then the function is a constant (whose value is easy to compute); otherwise the step is at $(1/z_i) (y_i - \vx_i' \cdot \mathbf{u})$, where the step being increasing or decreasing is determined by the sign of $z_i$ (and again the values of the function are easy to compute). It therefore follows that $\widehat{p}$ is a step function that has at most $|D| = O(N_\epsilon^d) = \poly(1 / \epsilon)$ steps, and that each of the values that the function takes can be computed in polynomial time with respect to $1 / \epsilon$ and $n$.
\end{proof}

\begin{proof}[Proof of~\Cref{thm:derandGaussian}]
    We give a recursive algorithm. If $n = 0$ then there is nothing to output, so assume $n \geq 1$. Let $\va = (t, \va')$, and $\vx_{ij} = (z_{ij}, \vx_{ij}')$ --- in other words, separate out the first variable. We are given that
    \begin{equation*}
    \sum_{i =1 }^m \Pr_{\mathbf{a}', t}\left[ \bigwedge_{j = 1}^d \vx_{ij}' \cdot \va' + z_{ij} t > y_{ij} \right] =
    \sum_{i =1 }^m \Pr_{\mathbf{a}}\left[ \bigwedge_{j = 1}^d \vx_{ij} \cdot \va > y_{ij} \right] \geq \rho,
    \end{equation*}
    when $\va = (\va', t) \sim \mathcal{N}(\mathbf{0}, \mathbf{I}_d)$. But then this must be true for some particular value of $t$, say $t^*$, i.e.
    \begin{equation}\label{eq:alpha}
    \sum_{i =1 }^m \Pr_{\va'}\left[ \bigwedge_{j = 1}^d \vx_{ij}' \cdot \va' + z_{ij} t^* > y_{ij} \right]  \geq \rho.
    \end{equation}
    Apply~\Cref{thm:stepApprox} to each of the probabilities above viewed as functions of $t$, with $\epsilon' = \epsilon / 2nm$; we thus build step functions $p_1, \ldots, p_m$ in polynomial time with respect to $n$ and $1/\epsilon' = 2nm/\epsilon$, such that
    \[
    \left| p_i(t) -  \Pr_{\va'}\left[ \bigwedge_{j = 1}^d \vx_{ij}' \cdot \va' + z_{ij} t > y_{ij} \right] \right| \leq \frac{\epsilon}{2nm}
    \]
    for all $t$.
    Add these equations for $i = 1, \ldots, m$ to find
    \begin{equation}\label{eq:err}
    \left| \sum_{i = 1}^m p_i(t) -  \sum_{i = 1}^m \Pr_{\va'}\left[ \bigwedge_{j = 1}^d \vx_{ij}' \cdot \va' + z_{ij} t > y_{ij} \right] \right| \leq \frac{\epsilon}{2n}.
    \end{equation}
    Observe that $\sum_{i = 1}^m p_i$ is a step function with polynomially many steps with respect to $n, m$, whose values are also computable in polynomial time. Thus it is easy to find some value $\hat{t}$ that maximises the expression $\sum_{i = 1}^m p_i(\hat{t})$. By~\eqref{eq:alpha} and~\eqref{eq:err} we have that $\sum_{i = 1}^m p_i(\hat{t}) \geq \sum_{i = 1}^m p_i(t^*) \geq \rho - \frac{\epsilon}{2n}$, and by~\eqref{eq:err} again we find that
    \[
    \sum_{i = 1}^m \Pr_{\va'}\left[ \bigwedge_{j = 1}^d \vx_{ij}' \cdot \va' + z_{ij} \hat{t} > y_{ij} \right] \geq \rho - \frac{\epsilon}{n}.
    \]
    Equivalently,
    \[
    \sum_{i = 1}^m \Pr_{\va'}\left[ \bigwedge_{j = 1}^d \vx_{ij}' \cdot \va' > y_{ij} - z_{ij} \hat{t} \right] \geq \rho - \frac{\epsilon}{n},
    \]
    and we can recursively find optimal values for the remaining random variables in $\va'$. Observe that our recursive depth is $n$, that at each level we use polynomial time with respect to $n, m, 1 / \epsilon$, and that, finally, at each step we lose $\epsilon / n$ from the sum of our probabilities. These facts together imply the correctness of our general derandomisation procedure.

    We note in passing that the total time complexity of our method is exponential in $d$; however this does not matter, as we consider $d$ a constant.
\end{proof}

This is enough to derandomise our algorithm.

\begin{proof}[Proof of~\Cref{thm:k2k3}]
    Let $G=(V,E)$, where $V = [n]$ (without loss of generality) and $m = |E|$.
    Suppose this graph has a cut of size $\rho$. By the analysis of our
    randomised algorithm from~\Cref{thm:main} with $\epsilon=1/3$, using SDP we can find, in polynomial time with respect to $G$, a set of vectors $\vx_1, \ldots, \vx_n$ such that $\norm{ \vx_i }^2 = 1$ and, if $\va_1, \va_2, \va_3 \sim \mathcal{N}(\mathbf{0}, \mathbf{I}_n)$ are normally distributed variables, then
    \[
    \sum_{(u, v) \in E} \Pr_{\va_1, \va_2, \va_3} \left[ \arg \max_{i} \va_k \cdot \vx_u \neq 
    \arg \max_i \va_i \cdot \vx_v \right] \geq \rho - \frac{1}{3}.
    \]
    Now, let $\va = (\va_1, \va_2, \va_3) \sim \mathcal{N}(\mathbf{0}, \mathbf{I}_{3n})$, and define $\vx_{ui}$ such that $\va \cdot \vx_{ui} = \va_i \cdot \vx_u$; in other words, pad out $\vx_u$ with $2n$ zeroes. In what follows, let $\oplus$ denote addition mod 3 over $[3]$. We first claim that the event
    \[
    \arg \max_{i} \va_i \cdot \vx_u \neq \arg \max_i \va_i \cdot \vx_v,
    \]
    can be seen as the disjoint union of 6 intersections of 4 half spaces in the space of $\va$. To express it in this way, first fix the value of the respective sides to $c \neq c'$, where $c, c' \in [3]$. Observe that the event that $\arg \max_i \va_i \cdot \vx_u = c$ is the same as $\va_c \cdot \vx_u > \va_{c\oplus 1} \cdot \vx_u \land \va_c \cdot \vx_u > \va_{c \oplus 2} \cdot \vx_u$. Now, using the notation from before, this is equivalent to $\va \cdot (\vx_{uc} - \vx_{u, c\oplus 1}) > 0 \land
    \va \cdot (\vx_{uc} - \vx_{u, c\oplus 2}) > 0$. It follows that
    \begin{multline*}
    \rho - \frac{1}{3} \leq 
    \sum_{(u, v) \in E} \Pr_{\va_1, \va_2, \va_3} \left[ \arg \max_{i} \va_i \cdot \vx_u \neq 
    \arg \max_i \va_i \cdot \vx_v \right] \\
    =
    \sum_{(u, v) \in E} \sum_{c \neq c'} \Pr_{\va} [
    \va \cdot (\vx_{uc} -
    \vx_{u,c\oplus 1}) > 0
   \land 
    \va \cdot (\vx_{uc} 
    - \vx_{u,c\oplus 2}) > 0
    \\
   \land 
    \va \cdot (\vx_{vc'} -
    \vx_{v,c' \oplus 1}) > 0
   \land 
    \va \cdot (\vx_{vc'}-
     \vx_{v,c'\oplus 2}) > 0
    ].
    \end{multline*}
    By~\Cref{thm:derandGaussian} with $\epsilon = \frac{1}{3}$, we can find particular values $\va^*$ such that 
    \begin{multline*}
    \sum_{(u, v) \in E} \sum_{c \neq c'} [
    \va^* \cdot (\vx_{uc} -
    \vx_{u,c\oplus 1}) > 0
   \land 
    \va^* \cdot (\vx_{uc} 
    - \vx_{u,c\oplus 2}) > 0
    \\
   \land 
    \va^* \cdot (\vx_{vc'} -
    \vx_{v,c' \oplus 1}) > 0
   \land 
    \va^* \cdot (\vx_{vc'}-
     \vx_{v,c'\oplus 2}) > 0
    ] \geq \\
    \geq \rho - \frac{1}{3} - \frac{1}{3} = \rho - \frac{2}{3}.
    \end{multline*}
    Defining $(\va_1^*, \va_2^*, \va_3^*) = \va^*$, this is equivalent to
    \[
    \sum_{(u, v) \in E} [\arg \max_i \va_i^* \cdot \vx_u \neq
    \arg \max_i \va_i^* \cdot \vx_v] \geq \rho - \frac{2}{3}.
    \]
    In other words, if we set the colour of vertex $u$ to $\arg \max_i \va_i^* \cdot \vx_u$, then we will correctly colour at least $\rho - 2 / 3$ edges. Now, note that $\rho \in \mathbb{N}$. Since the number of correctly coloured edges must also be an integer, and it is at least $\rho - 2 / 3$, it follows that it is at least $\rho$. Thus our algorithm returns a 3-colouring of value $\rho$, as required.
\end{proof}

\section{\texorpdfstring{Hardness of $\maxPCSP(K_2, \G_3)$}{Hardness of maxPCSP(K2,G3)}}
\label{sec:hardness}

In this section, we will prove the hardness part of the following result.

\main*

Our general strategy will be to gadget reduce from the 3-bit PCP of H{\aa}stad ~\cite{Hastad01}, similarly to~\cite{Trevisan00:sicomp} or~\cite{BGS:98}. The main difficulty comes in from the fact that it is not possible to ``negate'' variables in an obvious way, since ``negation'' is not globally preserved by the property of being triangle-free, as opposed to that of being bipartite. Some mild complications will be forced by this. Recall first the definition of exactly-3 linear equations.

\begin{definition}
    In the problem $\elin_\delta$, one is given a system of mod-2 linear
    equations with exactly 3 variables per equation; i.e.~$x + y + z \equiv 0
    \bmod 2$ or $x + y + z \equiv 1 \bmod 2$. If it is possible to
    simultaneously solve a $1 - \delta$ fraction of all the equations, one must
    answer \textsc{Yes}; otherwise, if it is not even possible to simultaneously solve a $\frac{1}{2} + \delta$ fraction of the equations, one must answer \textsc{No}.
\end{definition}

\begin{theorem}[\cite{Hastad01}]
    For every small enough $\delta$, the problem $\elin_\delta$ is \NP-hard.
\end{theorem}

To deal with our negation problems, we will need a ``balanced'' version of this problem.

\begin{definition}
    In the problem $\belin_\delta$, one is given a system of mod-2 linear equations with exactly 3 variables per equation; i.e.~$x + y + z \equiv 0 \bmod 2$ or $x + y + z \equiv 1 \bmod 2$. Furthermore, the number of equations of the two types is equal. 
    A \emph{balanced solution} to such a system of equations is one that satisfies exactly as many equations of form $x + y + z \equiv 0 \bmod 2$ as those of form $ x + y + z \equiv 1 \bmod 2$.
    
    If it is possible to find a balanced solution that satisfies a $1 - \delta$ fraction of all the equations, one must answer \textsc{Yes}; otherwise, it if is not even possible to find any (\emph{possibly even unbalanced}) solution that satisfies a $\frac{1}{2} + \delta$ fraction of the equations, one must answer \textsc{No}.
\end{definition}

We believe 
that~\cite{Hastad01} proves, without being explicit about it, \NP-hardness of $\belin_\delta$, although it is not straightforward to see it from the proof in~\cite{Hastad01}. For completeness, we provide a simple, self-contained reduction.

\begin{lemma}\label{lem:source}
    For every small enough $\delta$, the problem $\belin_\delta$ is \NP-hard.
\end{lemma}

\begin{proof}
    We reduce from $\elin_\delta$ to $\belin_\delta$.

    \paragraph{Reduction.}
    Given a system of $m$ equations $\mathcal{E}$ on $n$ variables, which contains the equations $x_i + y_i + z_i \equiv p_i \bmod 2$ for $i \in [m]$,\footnote{As usual, we denote by $[m]$ the set $\{1,\ldots,m\}$.} we define the system of equations $\mathcal{E}'$ on $n$ variables, which contains the equations $x_i' + y_i' + z_i' \equiv 1 - p_i \bmod 2$ for $i \in [m]$. We then return the system of equations $\mathcal{E} \sqcup \mathcal{E}'$ i.e.~the disjoint union of the two systems.

    \paragraph{Completeness.} Suppose that $\mathcal{E}$ has a solution $x_i \mapsto c(x_i)$ that satisfies a $1 - \delta$ fraction of the equations. Then the system $\mathcal{E} \sqcup \mathcal{E}'$ has a balanced solution that also satisfies a $1 - \delta$ fraction of its equations, namely the one that sends $x_i$ to $c(x_i)$ and $x_i'$ to $1- c(x_i)$. This solution is balanced since every equation $x_i + y_i + z_i \equiv p_i \bmod 2$ that it solves within $\mathcal{E}$ can be paired up with an equation $x_i' + y_i' + z_i' \equiv 1 - p_i \bmod 2$ which is solved in $\mathcal{E}'$.

    \paragraph{Soundness.} Suppose $\mathcal{E} \sqcup \mathcal{E}'$ has a
    solution $x_i \mapsto c(x_i)$, $x_i' \mapsto d(x_i')$, which satisfies a $\frac{1}{2} + \delta$ fraction of the equations. Such a solution must either satisfy a $\frac{1}{2} + \delta$ fraction of the equations within $\mathcal{E}$, or a $\frac{1}{2} + \delta$ fraction of the equations within $\mathcal{E}'$. In the first case, $c$ is the required solution for the original problem; in the second, $x_i \mapsto 1 - d(x_i')$ is the required solution.
\end{proof}

We now define the notion of ``gadget'' that we will need for this particular reduction. This is along the same lines as~\cite{BGS:98,Trevisan00:sicomp}, but (i) generalised to deal with promise problems and (ii) specialised to our particular promise problem.

For the following, if $G = (V, E)$ is any bipartite graph, and $V' \subseteq V$, then we say that a function $c : V' \to \{ 0, 1 \}$ is compatible with $G$ if it is possible to extend $c$ into a 2-colouring of $G$.

\begin{definition}\label{def:gadget}
    A gadget with performance $\alpha \in \N$ and parity $p \in \{0, 1\}$ is a graph $G = (V, E)$, with $0, x, y, z \in V$, where the following hold.
    \begin{enumerate}
        \item\label{item:soundness} For any function $c : \{ 0, x, y, z \} \to \{0, 1\}$ such that $c(0) + c(x) + c(y) + c(z) \equiv p \bmod 2$, there exists a bipartite subgraph $H$ of $G$ with $\alpha$ edges, such that $c$ is compatible with $H$.
        \item\label{item:completenes1} Any triangle-free subgraph of $G$ has at most $\alpha$ edges.
        \item\label{item:completenes2} Every triangle-free subgraph $H$ of $G$ with strictly more than $\alpha-1$ edges (and hence, by \Cref{item:completenes1}, exactly $\alpha$ edges) puts $0, x, y, z$ in the same connected component $C$, and the distance from $0$ to $x, y, z$ respectively is at most 2. Furthermore $C$ is bipartite, and for any $c : \{0, x, y, z\} \to \{0, 1\}$ that is compatible with $C$, we have that $c(0) + c(x) + c(y) + c(z) \equiv p \bmod 2$.
    \end{enumerate}
\end{definition}

\begin{lemma}\label{lem:bins}
    Suppose we have $n$ containers with capacities $c_1, \ldots, c_n \geq 0$. Suppose we distribute a volume of at least $c_1 + \cdots + c_n - n + a$ among the containers, distributing $v_i \leq c_i$ volume to container $i$. Then the number of containers $i$ for which $v_i > c_i - 1$ is at least $a$.
\end{lemma}
\begin{proof}
    We prove this by contrapositive: if we distribute $> c_i - 1$ volume to less than $a$ containers, then we cannot distribute $c_1 + \cdots + c_n - n + a$ or more.

    By permuting containers, we may assume we have distributed a volume $v_i > c_i - 1$ to containers $i = 1, \ldots, a - 1$, and a volume of $v_i \leq c_i - 1$ to containers $i = a, \ldots, n$. Conditional on this, each container is independent, and we see that the maximum volume that can be distributed is $c_1 + \cdots + c_{a - 1}$ in the first $a - 1$ containers, and $(c_a - 1) + \cdots + (c_n - 1) = c_a + \cdots + c_n - n + a - 1$. Hence the total amount of volume we could distribute is at most $c_1 + \cdots + c_n - n + a - 1 < c_1 + \cdots + c_n - n + a$.
\end{proof}

The next theorem encodes our reduction from 
the 3-bit PCP of~\cite{Hastad01} to $\maxPCSP(K_2, \G_3)$. This reduction is standard, needing only some care to deal with the fact that the triangle-free graph selected in the soundness case must be ``connected enough''.

\begin{theorem}
    Suppose that for $i \in \{0, 1\}$ there exist gadgets $G_i$ with performance $\alpha_i$ and parity $i$. Then for every $\epsilon > 0$ it is \NP-hard to approximate $\maxPCSP(K_2, \G_3)$  with approximation ratio $1 - 1 / (\alpha_0 + \alpha_1) + \epsilon$.
\end{theorem}
\begin{proof}
    We reduce $\belin_\delta$
    to approximate 
    $\maxPCSP(K_2, \G_3)$ with approximation ratio $1 - 1 / (\alpha_0 + \alpha_1) + \epsilon$. Our choice of $\delta$ will be bounded above by a value that depends on $\epsilon, \alpha_0, \alpha_1$. We now describe our gadget reduction.

    \paragraph{Reduction.} Suppose we are given an instance of $\belin_{\delta}$, with $2m$ constraints and $n$ variables $V$. Suppose that the constraints are $x_i^0 + y_i^0 + z_i^0 \equiv 0 \bmod 2$ and $x_i^1 + y_i^1 + z_i^1 \equiv 1 \bmod 2$ for $i \in [m]$.
     Our reduction first creates $n + 1$ vertices, one for every variable, plus a special vertex denoted by $0'$. For every constraint $x_i^p + y_i^p + z_i^p = p \bmod 2$ we create a copy of $G_{p}$, identifying vertices $0, x, y, z$ with $0', x_i^p, y_i^p, z_i^p$. Finally set $\rho = \lceil m(1 - \delta)(\alpha_0 + \alpha_1)\rceil$.

    \paragraph{Completeness.} Suppose that our $\belin{}$ instance has a balanced solution that satisfies at least a $1 - \delta$ fraction of the constraints. We claim that the graph outputted by our algorithm has a cut with at least $\lceil m (1 - \delta) (\alpha_0 + \alpha_1) \rceil$ edges. Indeed, this cut is guaranteed by~\Cref{item:soundness} in~\Cref{def:gadget}: place the vertices corresponding to variables in the original problem in the cut according to their value; then the vertices in the gadgets are placed according to the cut guaranteed by this assumption. The cut then has at least $m(1 - \delta)(\alpha_0 + \alpha_1)$ edges because the solution is guaranteed to be balanced. Since any cut must have an integral number of edges, this cut must even have $\lceil m (1 - \delta) (\alpha_0 + \alpha_1) \rceil$ edges.

    \paragraph{Soundness.} Suppose that the graph outputted by our reduction has a triangle-free subgraph $S$ with at least
    \[
    \lceil
    m (\alpha_0 + \alpha_1 - 1 + 2\delta)
    \rceil \geq
    m (\alpha_0 + \alpha_1 - 1 + 2\delta) = m(\alpha_0 + \alpha_1) - 2m + m + 2m\delta\]
    edges. We call a constraint with parity $i$ ``good'' if the intersection of $S$ with the gadget corresponding to that constraint has strictly more than $\alpha_i- 1$ edges. Observe that at least $m + 2m\delta$ constraints must be good by~\Cref{lem:bins}:  the gadgets for the $2m$ constraints are the containers; by~\Cref{item:completenes1} from~\Cref{def:gadget}, $m$ of the capacities are $\alpha_0$, and $m$ are $\alpha_1$; finally, a constraint is ``good'' if its container is allocated strictly more volume than its capacity minus one.
    We now show how to create a solution to the original $\belin{}$ instance that satisfies all the good constraints: this solution then satisfies a $(m + 2m \delta) / 2m = 1/2 + \delta$ fraction of the constraints.

    To create our solution, consider the subgraph $S$. Find the shortest path from $0'$ to every variable, then set the variable to the parity of the length of that path. (If there is no path, then we can set that variable arbitrarily.) Consider now any good constraint: suppose it is $x + y + z \equiv p \bmod 2$. Suppose $S'$ is the intersection of $S$ with the gadget for this constraint. Since $S$ is triangle-free, $S'$ must be triangle-free; since the constraint is good, $S'$ has strictly more than $\alpha_i - 1$ edges. By~\Cref{item:completenes2} of~\Cref{def:gadget}, $S'$ is bipartite; furthermore $S'$ connects $0'$ to $x, y, z$ by paths of length at most 2. Thus the shortest path from $0'$ to $x, y, z$ \emph{in $S$ (as well as $S'$)} is also of length at most 2. Since $S$ is triangle-free, the parity of the length of the shortest path from $0'$ to $x, y, z$ in $S$ matches the parity of any path from $0'$ to $x, y, z$ in $S'$. (To see why: if this were not the case, then for one of $x, y, z$, there are two paths from $0'$ of different parities, both of length at most 2. The first path is the one that exists in $S$ by assumption, the second one is the one that exists in $S'$ by \Cref{item:completenes2} of \Cref{def:gadget}. This implies a triangle in $S$, a contradiction.)
    Colouring according to the sides of the bipartite graph $S'$ satisfies the constraint according to~\Cref{item:completenes2} of~\Cref{def:gadget}, and our colouring is the same as this. So all good constraints are satisfied.

    \paragraph{Hardness factor.} Let
    \begin{multline*}
    \alpha =
    \frac{m(\alpha_0 + \alpha_1 - 1 +  2\delta)}{\lceil m(1 - \delta)(\alpha_0 + \alpha_1)\rceil}
    \leq 
    \frac{m(\alpha_0 + \alpha_1 - 1 +  2\delta)}{m(1 - \delta)(\alpha_0 + \alpha_1)}
    =
    \left(
    \frac{\alpha_0 + \alpha_1 - 1 + 2\delta}{\alpha_0 + \alpha_1}
    \right)\left(\frac{1}{1 - \delta}\right)\\
    =
    \left(1 - \frac{1}{\alpha_0 + \alpha_1} + O(\delta)\right)(1 + O(\delta))
    = 1 - \frac{1}{\alpha_0 + \alpha_1} + O(\delta),
    \end{multline*}
    where the $O$ notation hides factors depending on $\alpha_0,\alpha_1$.
    Note that $\alpha$ can be made at most $1 - 1 / (\alpha_0 + \alpha_1) + \epsilon$ by setting $\delta$ small enough, as a function of $\epsilon, \alpha_0, \alpha_1$.
    
    So far, we have shown that given a $\belin{}$ instance with a balanced solution that satisfies a $1 - \delta$ fraction of the equations, our reduction yields a graph with a cut with at least $\lceil m(1 - \delta)(\alpha_0 + \alpha_1)\rceil = \rho$ edges; and if the graph we output has a triangle-free subgraph with at least $\lceil m(\alpha_0 + \alpha_1 - 1 + 2\delta)\rceil = \lceil \alpha \rho \rceil$ edges then the original instance must have had a solution that satisfies at least a $1 / 2 + \delta$ fraction of equations. It follows that $\maxPCSP(K_2, \G_3)$ is \NP-hard to approximate with approximation ratio $\alpha \leq 1 - \frac{1}{\alpha_0 + \alpha_1} + \epsilon$.
\end{proof}
 
We now exhibit the gadgets. The first gadget is identical to a gadget of Bellare, Goldreich and Sudan~\cite{BGS:98} (although our analysis is slightly more complicated). In~\cite{BGS:98}, this gadget is called ``PC-CUT'', defined immediately before~\cite[Claim~4.17]{BGS:98}. The second gadget is a generalisation of the first. Recall that the gadgets of~\cite{BGS:98} were improved in~\cite{Trevisan00:sicomp}, and indeed the results of~\cite{Trevisan00:sicomp} indicate a generic method to find optimal gadgets for finite-domain CSPs. We do not believe this approach directly applies to our case because the property of being triangle-free is not captured by any finite CSP template (indeed, $\G_3$ is infinite, and any homomorphism-equivalent structure must also be).

We will write our gadgets as graphs with non-negative integer weights for simplicity of presentation. These gadgets can then be implemented by adding edges multiple times.

\begin{lemma}
    There exists a gadget with performance 9 and parity 1.
\end{lemma}
\begin{proof}
    Consider the complete graph on $\{ 0, a, x, y, z\}$. Set $w(a) = 2, w(0) = w(x) = w(y) = w(z) = 1$, and give an edge $(i, j)$ weight $w(i, j) = w(i) w(j)$. We now show that this satisfies the conditions from~\Cref{def:gadget}.
    \begin{enumerate}
    \item Without loss of generality suppose $c(0) = 0$. We have two cases. First, suppose that $c(x) = c(y) = c(z) = 1$. Then our cut is $(\{ x, y, z\}, \{ a, 0 \})$ which has weight 9. Conversely, suppose without loss of generality that $c(x) = c(y) = 0, c(z) = 1$. Then our cut is $(\{ 0, x, y \}, \{ a, z \})$, which also has weight 9.
    \item Consider any triangle-free subgraph of the gadget. This graph is either bipartite or it is $C_5$. If it is bipartite, suppose it has parts $(A, B)$. Then
    \[
    \sum_{i \in A, j \in B} w(i, j) = \left(\sum_{i \in A} w(i)\right)\left(\sum_{j \in B} w(j) \right),
    \]
    and the optimal cut puts $a$ (with weight 2) and one other vertex on one side, and the other three vertices on the other side --- this cut has weight 9, as required (all other cuts have weight at most 8). On the other hand, if the triangle-free subgraph is $C_5$ it has weight at most $2 + 2 + 1 + 1 + 1 = 7 \leq 9$.
    \item Consider any triangle-free subgraph of the gadget. By the analysis in the previous item, the only triangle-free subgraphs with weight greater than 8 are isomorphic to $K_{2,3}$, which put $a$ and one of $0, x, y, z$ on one side, and the other three vertices on the other side. It is not difficult to check that all of these graphs connect $0$ to $x, y, z$ with paths of length at most 2, and that the sides of the cut exhibit the correct parity requirements.\qedhere
    \end{enumerate}
\end{proof}

\begin{lemma}
    There exists a gadget with performance 17 and parity 0.
\end{lemma}
\begin{proof}
    Consider the complete graph on $\{0, 1, a, x, y, z\}$. Let $w(a) = w(1) = 2$, $w(0) = w(x) = w(y) = w(z) = 1$, $\delta(0, 1) = \delta(1, 0) = 1$, whereas otherwise $\delta(i, j) = 0$. Then define $w(i, j) = w(i) w(j) + \delta(i, j)$.

    \begin{enumerate}
        \item Without loss of generality suppose $c(0) = 0$. There are two cases. First suppose $c(x) = c(y) = c(z) = 0$. Then our cut is $(\{x, y, z, 0\}, \{a, 1\})$ with weight 17. Otherwise, suppose without loss of generality that $c(x) = 0, c(y) = c(z) = 1$. Then our cut is $(\{ 1, y, z\}, \{ 0, a, x\})$, with weight 17.
        \item Consider all the triangle-free subgraphs of the gadget. This subgraph is either bipartite, or it is a subgraph of the following graph
\begin{center}
    \begin{tikzpicture}[scale=0.6]
        \node[draw, circle, very thick] (A) at (0:3cm) {};
        \node[draw, circle, very thick] (B) at (60:3cm) {};
        \node[draw, circle, very thick] (C) at (120:3cm) {};
        \node[draw, circle, very thick] (D) at (180:3cm) {};
        \node[draw, circle, very thick] (E) at (240:3cm) {};
        \node[draw, circle, very thick] (F) at (300:3cm) {};
        \draw[black, very thick] (A) -- (B) -- (C) -- (D) -- (E) -- (A) (D) -- (F) -- (A);
    \end{tikzpicture}
\end{center}
    First we deal with the bipartite case. The weight of a complete bipartite graph with parts $(A, B)$ is
    \[
    \left(\sum_{i \in A, j \in B} \delta(i, j)\right) + \left(\sum_{i \in A} w(i)\right)\left( \sum_{j \in B} w(j)\right).
    \]
    The first part of the sum is 1 if and only if $0$ and $1$ are in different sides of the cut, and the second is maximised exactly when the two sides of the cut are as even as possible i.e.~have weight 4 and 4 respectively --- cuts that satisfy both of these conditions weight at most $17 = 4 \times 4 + 1$ edges, while all other cuts have weight at most 16.
    
    Now consider a non-bipartite triangle-free subgraph. The largest this could be is the graph pictured above. First consider the contribution of $w(i) w(j)$ to the total weight. Suppose for contradiction that the contribution is $> 14$; since there are 7 edges, there must exist one edge whose contribution is $> 2$. Since $w(i) w(j) \in \{1, 2, 4\}$, and $w(i)w(j) = 4$ only for $\{i, j \} = \{ a, 1 \}$, $a$ and $1$ are adjacent within the subgraph. Since any way of placing $a, 1$ adjacently within the subgraph leads to them having at most 3 other neighbours, the contribution of $w(i) w(j)$ is at most $4 + 3 \times 2 + 3 \times 1 = 13 \not > 14$, a contradiction. So the contribution of $w(i) w(j)$ is at most 14, and the total weight is at most 15.

    \item Consider now any triangle-free subgraph of the gadget with weight greater than 16. By the analysis in the previous item, this subgraph must be a complete bipartite subgraph where the two parts both have weight 4, and 0 and 1 are put on opposite sides of the cut. Consider the side of $a$: if it is on the same side as $1$, then $x, y, z$ are all on the opposite side i.e.~together with 0. Conversely, if $a$ is on the same side as $0$, then two of $x, y, z$ must be on the same side as $1$, and the remaining vertex among $x, y, z$ must be on the same side as $0$. It is not difficult to check that these bipartite subgraphs satisfy the desired connectivity and parity conditions.\qedhere
    \end{enumerate}
\end{proof}

Note that we do \emph{not} claim these gadgets are optimal. Hence it may be possible to improve the hardness factor even more via a better crafted gadget. 

\section{\texorpdfstring{Hardness of $\maxPCSP(C_{2k+1}, K_\ell)$}{Hardness of maxPCSP(C2k+1, Kl)}}
\label{sec:ugchardness}

In this section, we will prove the following result.

\thmUGChardness*

The proof is a generalisation of the proof
in~\cite{Dinur09:sicomp} that establishes the \NP-hardness of \emph{almost
3-colouring} (i.e.~given an $n$ vertex graph $G$, output \textsc{Yes} if $G$ has
a 3-colourable $(1 - \epsilon)n$ vertex-induced subgraph, and \textsc{No} if $G$
does not even have an independent set with more than $\epsilon n$ vertices), assuming the \UGC.
We will need, unlike~\cite{Dinur09:sicomp}, a
\emph{multilayered unique games conjecture}, in the style of~\cite{BWZ21}. We
first set up the necessary ingredients. The proof is based on \emph{Markov-chain noise operators} --- the one we will use is intimately related to $C_{2k + 1}$.

\begin{definition}
    Define $M_{2k+1}$ to be the matrix which has $\frac{1}{2}$ at position $(i, j)$
    if and only if $|i - j| \equiv 1 \bmod 2k+1$. I.e.~$M_{2k+1}$ is the
    \emph{circulant matrix} given by the vector $(0, \frac{1}{2}, 0, \ldots, 0,
    \frac{1}{2})$, of length $2k+1$. Note that $M_{2k+1}$ is the transition matrix of a Markov chain on $[2k+1]$, and $(i, j)$ has nonzero transition probability if and only if $(i, j)$ is an edge of $C_{2k+1}$.
\end{definition}

\begin{lemma}
    The uniform distribution is the stationary distribution of $M_{2k+1}$.
\end{lemma}
\begin{proof}
    Easy to check, but also follows immediately since $M_{2k+1}$ is the random walk on $C_{2k+1}$, which is an undirected, connected, regular graph.
\end{proof}

\begin{lemma}
    The eigenvalues of $M_{2k+1}$ are $\cos(2i\pi / (2k+1))$ for $i = 0, \ldots, 2k$. In particular, this matrix has one eigenvalue with absolute value 1; all the others have absolute at most $\cos(1 - \pi / (2k + 1)) < 1$.
\end{lemma}
\begin{proof}
    This follows from the formula for the eigenvalues of a circulant matrix.
\end{proof}

These are the key properties needed to apply the theory of \cite{Dinur09:sicomp}. We will need a multi-layered Unique Games Conjecture, which we now state.

\begin{definition}\label{def:layered}
    An $\ell$-layered unique label-cover instance consists of a set of variables
    $X_1, \ldots, X_\ell$, a domain $[D]$, and a multiset of constraints. Each
    constraint consists of $\ell$ variables $(x_1, \ldots, x_\ell) \in X_1
    \times \cdots \times X_\ell$, together with a family of permutations
    $\pi_{ij}$ on $[D]$ for $1 \leq i < j \leq \ell$, such that $\pi_{ik} =
    \pi_{jk} \circ \pi_{ij}$. A solution is an assignment $c : (X_1  \cup \ldots
    \cup X_\ell) \to [D]$. The assignment $c$ \emph{strongly satisfies} a
    constraint given by $(x_1, \ldots, x_{\ell})$ and $(\pi_{ij})_{1 \leq i < j
    \leq \ell}$ if $\pi_{ij}(c(x_i)) = c(x_j)$ \emph{for all $1 \leq i < j \leq
    \ell$}. The assignment $c$ \emph{weakly satisfies} this same constraint if
    $\pi_{ij}(c(x_i)) = c(x_j)$ \emph{for at least one pair $1 \leq i < j \leq
    \ell$}. The strong value of an instance is the maximum fraction of
    constraints that can be simultaneously strongly satisfied by some
    assignment; the weak value is given by the maximum fraction of constraints
    that can be simultaneously weakly satisfied by some assignment.
\end{definition}

Note that for $\ell = 2$, weak satisfaction and strong satisfaction (and hence weak and strong values) coincide. Hence for $\ell = 2$ we drop the weak/strong distinction.

\begin{conjecture}[\UGC~\cite{Khot02stoc}]\label{conj:UGC}
For every $\epsilon$ there exists $D$ such that, given a 2-layered unique label-cover instance with domain $[D]$, it is \NP-hard to distinguish if the value is at least $1 - \epsilon$ or not even $\epsilon$.
\end{conjecture}

We will show that \Cref{conj:UGC} implies the following conjecture. 

\begin{conjecture}[Multilayered UGC]\label{conj:mulUGC}
For every $\epsilon, \ell \geq 2$ there exists some $D$ such that, given an $\ell$-layered unique label-cover instance with domain $[D]$, it is \NP-hard to distinguish if the strong value is at least $1 - \epsilon$, or if the weak value is not even $\epsilon$.
\end{conjecture}

\begin{proof}[{Proof of \Cref{conj:mulUGC} from \Cref{conj:UGC}}]

Our reduction is essentially identical to that of~\cite{BWZ21}, which in turn builds on~\cite{DinurGKR05}. Our proof differs from the proof in~\cite{BWZ21} in two ways: (i) we must show that this reduction, if given a unique label cover, produces a multilayered label cover with bijective constraints, and (ii) we need to take care of the completeness case, which is imperfect here unlike in~\cite{BWZ21}. The proof of soundness is identical to that in~\cite{BWZ21}.

\paragraph{Reduction.} Suppose we are given a 2-layered unique label-cover instance with sets of variables $A, B$, domain $D$, and constraints $(a, b) \in E$ with permutations $\pi_{ab}$ for $(a, b) \in E$. Our reduction produces an $\ell$-layered unique label-cover instance, with sets of variables $X_1, \ldots, X_\ell$, where $X_i = B^{i - 1}\times A^{\ell - i}$, on domain $[D]^{\ell - 1}$. For every $(\ell - 1)$-tuple of constraints $(a_1, b_1), \ldots, (a_{\ell -1 }, b_{\ell -1 }) \in E$, where we let $\pi_i = \pi_{a_i, b_i}$, we create a constraint on variables $(x_1, \ldots, x_{\ell})$, where
\[
x_i = (b_1, \ldots, b_{i - 1}, a_i, \ldots, a_{\ell - 1}).
\]
Now, we define $\pi_{ij}$ as follows. Observe that $x_i$ and $x_j$ share the first $i - 1$, as well as the last $\ell - j$ variables in common --- the variables on the indices in between (i.e.~$i$ up to $j - 1$) are different. Hence we write
\[
\pi_{ij}(d_1, \ldots, d_{\ell - 1})
=
(d_1, \ldots, d_{i - 1}, \pi_{i}(d_{i}), \ldots, \pi_{j - 1}(d_{j - 1}), d_{j}, \ldots, d_{\ell - 1}).
\]
Observe that every such constraint is a bijection on $[D]^{\ell - 1}$ --- indeed,
\[
\pi_{ij}^{-1}(d_1, \ldots, d_{\ell - 1})
=
(d_1, \ldots, d_{i - 1}, \pi_{i}^{-1}(d_{i}), \ldots, \pi_{j - 1}^{-1}(d_{j - 1}), d_{j}, \ldots, d_{\ell - 1}).
\]
Furthermore, it is easy to check that $\pi_{ik} = \pi_{jk} \circ \pi_{ij}$ for every $1 \leq i < j < k \leq \ell$. Noting that the number of vertices and edges is polynomial completes the reduction.

\paragraph{Completeness.} Suppose that the original instance has value $1 -
  \delta$. We need to show that we can select $\delta$ small enough, independent of $D$, so that the final instance has strong value $1 - \epsilon$, for every $\epsilon, \ell$. Let $c$ be the solution that witnesses the value of the original instance. We claim that $c'$ given by
\[
c'(x_1, \ldots, x_{\ell - 1}) = (c(x_1), \ldots, c(x_{\ell - 1}))
\]
has strong value at least $(1 - \delta)^{\ell - 1}$. Setting $\delta = 1 -
  \sqrt[\ell - 1]{1 - \epsilon}$ makes this value equal $1 - \epsilon$, as
  required. To see why this is the case, suppose we select a constraint of the
  output instance uniformly at random. We want to show that this constraint is
  satisfied with probability at least $(1 - \delta)^{\ell - 1}$. Observe that,
  by construction, selecting a constraint of the output uniformly at random is
  the same as selecting an $(\ell - 1)$-tuple of constraints $(a_1, b_1),
  \ldots, (a_{\ell - 1}, b_{\ell -1}) \in E$ uniformly and independently at
  random. Each of these is satisfied by $c$ with probability $1 - \delta$; hence
  all constraints are satisfied by $c$ with probability $(1 - \delta)^{\ell - 1}$. Furthermore, when $c$ satisfies all the constraints $(a_i, b_i)$ then 
 $c'$ satisfies the constraint in the output instance built from these constraints. This completes the proof.

\paragraph{Soundness.} Identical to the proof in~\cite{BWZ21}. (This requires that the \UGC{} instance we start from is bi-regular, but this is possible by a standard transformation, cf.~\cite{Khot12:ccc-survey} or~\cite{Arora09:book,DinurGKR05}.)
\end{proof}

The proof of~\Cref{thm:ugcHardness} will be based on the \emph{long code
construction}~\cite{BGS:98}. We will now describe the building blocks of our reduction.

\begin{definition}
    Fix $k, \ell, D$. A \emph{cloud} of vertices, denoted by $\vec{f}$, is a set
    of vertices $f(a_1, \ldots, a_D)$ for $a_1, \ldots, a_D \in [2k + 1]$. For
    $\ell$ clouds of variables $\vec{f}_1, \ldots, \vec{f}_\ell$ and a family of
    permutations $\pi_{ij} : [D] \to [D]$ for $1 \leq i < j \leq \ell$ as in
    \Cref{def:layered}, we define the set of edges $E_\pi(\vec{f}_1, \ldots,
    \vec{f}_\ell)$ as follows: for every $1 \leq i < j \leq \ell$, $a_1, \ldots,
    a_D, b_1, \ldots, b_D \in [2k + 1]$ and for which every $a_t, b_t$ differ by $\pm 1$ modulo $2k + 1$, we include the edge $f_i(a_{\pi_{ij}(1)}, \ldots, a_{\pi_{ij}(D)}) - f_j(b_1, \ldots, b_D)$.
\end{definition}

The following is the key theorem from~\cite{Dinur09:sicomp} that we will use. We will define some notions, however we refer to~\cite{Dinur09:sicomp} for a full treatment.

\begin{definition}
    For a symmetric Markov operator $T$ on $[q]$, and letting $f, g : [q]^D \to \mathbb{R}$, the value $\langle f, T^{\otimes D} g \rangle$ has the following interpretation. Let $x \in [q]^D$ be distributed uniformly at random, and let $y \in [q]^D$ be such that $y_i$ is distributed according to the transition probabilities in $T$ starting at $x_i$. Then $\langle f, T^{\otimes D} g \rangle$ is the expected value of $f(x) g(y)$.

    The quantity $\Gamma_\rho(\mu, \nu)$, has the following interpretation. Let $x, y$ be two normally distributed variables with mean 0, variance 1 and covariance $\rho$. Then this value is the probability that $x \leq \Phi^{-1}(\mu), y \geq \Phi^{-1}(1 - \nu)$, where $\Phi$ is the cumulative distribution function of the normal distribution. Essentially, this value is nondecreasing in both $\mu$ and $\nu$.

    For a function $f : [q]^D \to \mathbb{R}$ and an integer $t$, the value $I_i^{\leq t}(f)$ is the \emph{low-degree influence} of coordinate $i$ in $f$. In particular, if $f(x) \in [0, 1]$, it can be shown that $\sum_{i = 1}^D I_i^{\leq t}(f) \leq t$ and $I_i^{\leq t}(f) \geq 0$. Furthermore, the influence of a coordinate is defined compatibly with permuting coordinates i.e.~if the influence of coordinate $i$ is $x$, and we permute the coordinates of $f$ so as to move $i$ to position $j$, yielding a function $g$, then the influence of $j$ in $g$ is still $x$.
\end{definition}

\begin{theorem}[\cite{Dinur09:sicomp}]\label{thm:dinur_thm}
    Fix $q$ and let $T$ be a symmetric Markov operator on $[q]$ with spectral radius $\rho < 1$ (by spectral radius we mean the second largest eigenvalue of $T$ in absolute value). Then for any $\epsilon > 0$ there exist $\delta > 0$ and $t \in \mathbb{N}$ so that if $f, g : [q]^D \to [0, 1]$ are two functions with
    \[
    \min(I_i^{\leq t}(f), I_i^{\leq t}(g)) < \delta,
    \]
    for all $i$, then
    \[
    \langle f, T^{\otimes D} g\rangle \geq \Gamma_\rho(\mu, \nu) - 
    \epsilon,
    \]
    where $\mu = E[f], \nu = E[g]$.
\end{theorem}

Similarly to \cite[Corollary 4.12]{Dinur09:sicomp}, we will use this in the contrapositive, in particular in the following form.

\begin{corollary}\label{cor:technical}
    Fix $q$ and let $T$ be a symmetric Markov operator on $[q]$ with spectral
    radius $\rho < 1$. For every $\epsilon>0$ there exist $\delta > 0$ and $t \in \mathbb{N}$ so that, if $f, g : [q]^D \to [0, 1]$ are two functions with $E[f] \geq \epsilon, E[g] \geq \epsilon$ and $\langle f, T^{\otimes D} g\rangle \leq \delta$, then there exists $i \in [D]$ so that $I_i^{\leq t}(f) \geq \delta$ and $I_i^{\leq t}(g) \geq \delta$.
\end{corollary}
\begin{proof}
Apply \Cref{thm:dinur_thm}; the values $\epsilon', \delta'$ are the values we use for \Cref{thm:dinur_thm}. We take $\epsilon' = \min(\epsilon, \Gamma_\rho(\epsilon, \epsilon) / 2)$ and take $\delta < \min(\delta', \Gamma_\rho(\epsilon, \epsilon)/2)$.

Suppose that $\langle f, T^{\otimes D} g \rangle \leq \delta < \Gamma_\rho(\epsilon, \epsilon) / 2$ --- it suffices to prove that $\langle f, T^{\otimes D} g\rangle < \Gamma_\rho(\mu, \nu) - \epsilon'$, where $\mu = E[f], \nu = E[g]$. Since $\epsilon \leq \mu, \epsilon \leq \nu$ it is sufficient to show that $\langle f, T^{\otimes D} g \rangle < \Gamma_\rho(\epsilon, \epsilon) - \epsilon'$. Since $\epsilon' \leq \Gamma_\rho(\epsilon, \epsilon) / 2$ it is sufficient to show that $\langle f, T^{\otimes D} g \rangle < \Gamma_\rho(\epsilon, \epsilon) / 2 $ --- which is true by our choice of $\delta$.
\end{proof}

\begin{lemma}\label{lem:technical}
    There exist $s, \delta$ which depend only on $k, \ell$ so that the following holds for any $D$.

    Consider an $\ell$-colouring $\vec{f} \to [\ell]$ of the cloud of vertices $\vec{f}$. There exists a way to assign any such cloud a subset $I(\vec{f})$ of $[D]$ of size $s$ such that the following holds.
    
    Consider any $\ell + 1$ clouds $\vec{f}_1, \ldots, \vec{f}_{\ell + 1}$ and a family of permutations $\pi_{ij}$ as in \Cref{def:layered}. Suppose that these vertices are $\ell$-coloured, and that the $\ell$-colouring satisfies a $(1 - \delta)$-fraction of the edges in $E_\pi(\vec{f}_1, \ldots, \vec{f}_{\ell + 1})$. Then there exist $1 \leq i < j \leq \ell + 1$ such that $\pi_{ij}(I(\vec{f}_i)) \cap I(\vec{f}_j) \neq \emptyset$.
\end{lemma}

\begin{proof}
    Consider some cloud $\vec{f}$. Note that this cloud must have a most frequent colour; let $c$ be that colour, and let $f : [2k + 1]^D \to \{0, 1\}$ be the indicator function of this colour inside $\vec{f}$ i.e.~$f(a_1, \ldots, a_D)$ is 1 if $\vec{f}(a_1, \ldots, a_D)$ is colored by $c$ and 0 otherwise. To construct $I(\vec{f})$, apply \Cref{cor:technical} to $f$ with $T = M_{2k + 1}$ and $\epsilon'= \frac{1}{\ell}$; suppose we get $t$ and $\delta'$ from \Cref{cor:technical}. We set $\delta = \frac{\delta'}{(\ell + 1)^2}$ and 
    \[
    I(\vec{f}) = \{ i \in [D] \mid I_i^{\leq t}(f) \geq \delta \}.
    \]
    Recall that $\sum_{i = 1}^D I_i^{\leq t}(f) = t$ and $I_i^{\leq t}(f) \geq 0$. So $|I(\vec{f})| \leq \frac{t}{\delta}$, and we set $s = \frac{t}{\delta}$.

    Now we must prove that the chain condition holds. Consider clouds $\vec{f}_1, \ldots, \vec{f}_{\ell + 1}$ and a family of permutations $\pi_{ij}$ as in \Cref{def:layered}. By the pigeonhole principle, there exist $i < j$ such that the most frequent colour in $\vec{f}_i, \vec{f}_j$ coincides. Say that colour is $c$, and let $\vec{f} = \vec{f}_i, \vec{g} = \vec{g}_i, \pi = \pi_{ij}$ and let $f, g : [2k + 1]^{D} \to \{0, 1\}$ be the indicator functions of the colour $c$ in $\vec{f}, \vec{g}$.

    We define the function $f^{\pi}$ as follows:
    \[
    f^\pi(x_1, \ldots, x_D) = f(x_{\pi(1)}, \ldots, x_{\pi(D)}).
    \]
    Note that $\pi(I(f)) = I(f^{\pi})$ since influences are compatible with permuting coordinates. Hence by \Cref{lem:technical} it suffices to show that $\langle f^\pi, M_{2k + 1}^{\otimes D} g\rangle \leq \delta' = (\ell + 1)^2 \delta$. Suppose for contradiction that $\langle f^\pi, M_{2k + 1}^{\otimes D} g \rangle > (\ell + 1)^2 \delta$.

    Recall that $\langle f^\pi, M_{2k + 1}^{\otimes D} g \rangle$ can be defined equivalently as follows. Suppose $x_1, \ldots, x_D$ are drawn uniformly and independently at random from $[2k + 1]$, and $y_1, \ldots, y_D$ is drawn so that $y_i$ follows the distribution of one step of $M_{2k + 1}$ starting from $x_i$. In other words, $y_i = (x_i \pm 1) \bmod 2k + 1$, with the $\pm$ being $+$ or $-$ with probability $1/2$. Then $\langle f^\pi, M_{2k + 1}^{\otimes D} g \rangle$ is the expected value of $f(x_{\pi(1)}, \ldots, x_{\pi(D)}) g(y_1, \ldots, y_D)$. But, by construction, this value is equal to the fraction of edges in $E_\pi(\vec{f}_1, \ldots, \vec{f}_{\ell + 1})$ between $\vec{f}$ and $\vec{g}$ that have both endpoints coloured to $c$. Since the number of edges between $\vec{f}_i$ and $\vec{f}_j$ in $E_\pi(\vec{f}_1, \ldots, \vec{f}_{\ell + 1})$ is equal for every choice of $i$ and $j$, of which (coarsely) there are at most $(\ell + 1)^2$, it follows that a $> \delta$ fraction of the edges in $E_{\pi}(\vec{f}_1, \ldots, \vec{f}_{\ell + 1})$ have both endpoints coloured to $c$. This contradicts our assumptions.
\end{proof}

We will now prove~\Cref{thm:ugcHardness}, which we restate here for the reader's
convenience.

\thmUGChardness*
 
\begin{proof}
  Fix $k \geq 1, \ell \geq 3$. We will reduce from~\Cref{conj:mulUGC}, which is implied by the \UGC, to
  $\maxPCSP(C_{2k + 1}, K_{\ell})$ in the case where $\rho = 1 - \epsilon$ for some very small $\epsilon$.
  We will show that there exists some completeness/soundness for \Cref{conj:mulUGC}, call it $\epsilon'$, for which our reduction will work.
    
\paragraph{Reduction.} Suppose we are given an instance of $(\ell + 1)$-layered unique label-cover with domain $D$. Suppose the variable set is $X_1, \ldots, X_{\ell + 1}$, and the constraint set $E \subseteq X_1 \times \cdots \times X_{\ell + 1}$. For every variable $x$, we introduce a cloud of vertices $\vec{f}_x$. For every constraint $(x_1, \ldots, x_{\ell + 1}) \in E$, characterised by functions $(\pi_{ij})_{1 \leq i < j \leq \ell + 1}$, we introduce the set of edges $E_\pi(\vec{f}_{x_1}, \ldots, \vec{f}_{x_{\ell + 1}})$ --- if an edge is included in multiple copies, we introduce it with a higher multiplicity.

\paragraph{Completeness.} Suppose that the original instance has a strong value of $1 - \epsilon'$, witnessed by $c$. We claim that the reduced instance has a $(1-\epsilon')$-fractional $C_{2k + 1}$-colouring. Indeed let $c'$ be this colouring, and colour vertex $f_x(a_1, \ldots, a_D)$ by
\[
c'(f_x(a_1, \ldots, a_D)) = a_{c(x)}.
\]
Now, we must show that a $1 - \epsilon'$ fraction of edges are coloured properly. We sample the edge as follows: first sample a constraint in the original instance uniformly at random, $(x_1, \ldots, x_{\ell + 1})$, then sample some edge from $E_\pi(\vec{f}_{x_1}, \ldots, \vec{f}_{x_{\ell + 1}})$ uniformly at random. Since each constraint spawns the same number of edges, this is the same as sampling an edge uniformly at random. Now, we will show that if the constraint $(x_1, \ldots, x_{\ell + 1})$ was satisfied by $c$, then \emph{all} the edges in $E_\pi(\vec{f}_{x_1}, \ldots, \vec{f}_{x_{\ell + 1}})$ are satisfied --- this is sufficient to show completeness.

  Consider an arbitrary edge $f_{x_i}(a_{\pi_{ij}(1)}, \ldots, a_{\pi_{ij}(D)})
  - f_j(b_1, \ldots, b_D)$ from a constraint $(x_1, \ldots, x_\ell)$ that was
  satisfied in the original instance. The endpoints are coloured by
  $a_{\pi_{ij}(c(x_i))}$ and $b_{c(x_j)}$. Since the original constraint was
  satisfied, we know that $\pi_{ij}(c(x_i)) = c(x_j)$, hence these colours are
  just $a_{c(x_j)}$ and $b_{c(x_j)}$. By construction, we know that these two
  values differ by $\pm 1$ modulo $2k + 1$, and hence the edge is satisfied.

Hence for the completeness to work, we need to take $\epsilon' \leq \epsilon$.

\paragraph{Soundness.} Suppose we have a $(1-\epsilon)$-fractional $\ell$-colouring of the reduced graph. We must show that the original instance has a weak solution of value $\epsilon'$.
Consider the unsatisfied edges; suppose that we select a constraint $(x_1, \ldots, x_{\ell + 1})$ at random, and we count the expected number of unsatisfied edges is $E_\pi(\vec{f}_{x_1}, \ldots, \vec{f}_{x_{\ell + 1}})$. Noting that there are the same number of edges in $E_\pi(\vec{f}_{x_1}, \ldots, \vec{f}_{x_{\ell + 1}})$ for every constraint $(x_1, \ldots, x_{\ell + 1})$, by Markov's inequality, the probability that there are more than $\sqrt{\epsilon}$ unsatisfied edges is at most $\sqrt{\epsilon}$. Hence, for a $(1 - \sqrt{\epsilon})$-fraction of the constraints $(x_1, \ldots, x_{\ell + 1})$, a $(1 - \sqrt{\epsilon})$-fraction of the edges in $E_\pi(\vec{f}_{x_1}, \ldots, \vec{f}_{x_{\ell + 1}})$ are coloured correctly.

Now, apply \Cref{lem:technical} to get $s, \delta$ and an assignment function $I$. Take $\epsilon$ so that $\sqrt{\epsilon} \leq \delta$. We claim that assigning variable $x$ a value from $I(f_x)$ uniformly at random will weakly satisfy a $\frac{1 - \sqrt{\epsilon}}{s^2}$-fraction of the constraints. Taking $\epsilon' \leq \frac{1 - \sqrt{\epsilon}}{s^2}$ will make the reduction work. For each of the $(1 - \sqrt{\epsilon})$-fraction of the constraints where a $(1 - \sqrt{\epsilon})$ fraction of the edges in $E_\pi(\vec{f}_{x_1}, \ldots, \vec{f}_{x_{\ell + 1}})$ are coloured correctly,  we claim that the constraint is weakly satisfied with probability at least $1 / s^2$. Indeed, this follows from \Cref{lem:technical}: the edges in $E_\pi(\vec{f}_{x_1}, \ldots, \vec{f}_{x_{\ell+1}})$ satisfy the preconditions of \Cref{lem:technical}, so for some $1 \leq i < j \leq \ell + 1$ we have that $\pi_{ij}(I(\vec{f}_{x_i})) \cap I(\vec{f}_{x_j}) \neq \emptyset$. So, with probability $1 / s^2$ we choose the correct value for $x_i$ and $x_j$, and weakly satisfy the constraint.
\end{proof}

\section*{Acknowledgements}

We thank Jakub Opr\v{s}al, who asked us about the complexity of $\maxPCSP$s for
graphs beyond cliques, and for useful discussions. We also thank Pravesh Kothari
for asking us about the performance of our maximum bipartite vs.~triangle-free
subgraph algorithm in the $1 - \epsilon$ regime, which led to~\Cref{thm:almost}.
Finally, we thank the three anonymous reviewers for their feedback on this paper.

{
\bibliography{nz}
\bibliographystyle{alphaurl}
}

\end{document}